\documentclass[12pt]{amsart}
\usepackage[latin9]{inputenc}
\usepackage[letterpaper]{geometry}
\usepackage{color}
\usepackage{bm}
\usepackage{amsbsy}
\usepackage{amstext}
\usepackage{amsthm}
\usepackage{amssymb}
\usepackage{setspace}
\usepackage[authoryear]{natbib}
\usepackage[unicode=true,
 bookmarks=false,
 breaklinks=true,pdfborder={0 0 0},pdfborderstyle={},backref=false,colorlinks=true]
 {hyperref}
\hypersetup{
 linkcolor=Cerulean,citecolor=Cerulean}
\usepackage{breakurl}

\makeatletter

\InputIfFileExists{t2aenc.def}{}{%
  \errmessage{File `t2aenc.def' not found: Cyrillic script not supported}}
\DeclareRobustCommand{\cyrtext}{%
  \fontencoding{T2A}\selectfont\def\encodingdefault{T2A}}
\DeclareRobustCommand{\textcyr}[1]{\leavevmode{\cyrtext #1}}

\DeclareTextSymbolDefault{\textquotedbl}{T1}

\theoremstyle{definition}
\newtheorem{defn}{\protect\definitionname}
\theoremstyle{definition}
 \newtheorem{example}{\protect\examplename}
\theoremstyle{plain}
\newtheorem{fact}{\protect\factname}
\theoremstyle{plain}
\newtheorem{prop}{\protect\propositionname}
\theoremstyle{plain}
\newtheorem{lem}{\protect\lemmaname}
\theoremstyle{plain}
\newtheorem{cor}{\protect\corollaryname}
\theoremstyle{definition}
\newtheorem*{example*}{\protect\examplename}

\usepackage{amsfonts}
\usepackage{bm}
\usepackage{etex}
\usepackage{setspace}
\usepackage{refstyle}

\usepackage{fourier} 
\usepackage[scaled=0.875]{helvet} 

\usepackage{graphics}
\usepackage{epsfig}\usepackage{verbatim}\usepackage{bm}\usepackage{latexsym}\usepackage{url}\usepackage{rotating}

\usepackage[usenames,dvipsnames,svgnames,table]{xcolor}
\usepackage{hyperref}

\usepackage{mathrsfs}
\usepackage{multirow}
\usepackage{array}
\usepackage{url}

\allowdisplaybreaks

\newref{sec}{name=Section~}
\newref{subsec}{name=Section~}
\newref{thm}{name=Theorem~}
\newref{prop}{name=Proposition~}
\newref{lem}{name=Lemma~}

\usepackage{tikz}
\usetikzlibrary{arrows}
\usetikzlibrary{graphs}
\usetikzlibrary{shapes.geometric}
\usepackage{bm}
\usepackage{tikz-cd}

\makeatother

\providecommand{\corollaryname}{Corollary}
\providecommand{\definitionname}{Definition}
\providecommand{\examplename}{Example}
\providecommand{\factname}{Fact}
\providecommand{\lemmaname}{Lemma}
\providecommand{\propositionname}{Proposition}

\begin{document}

\title{Expectations, Networks, and Conventions}

\author{Benjamin Golub \and Stephen Morris}
\begin{abstract}
In coordination games and speculative over-the-counter financial markets,
solutions depend on \emph{higher-order average expectations}: agents'
expectations about what counterparties, on average, expect \emph{their}
counterparties to think, etc. We offer a unified analysis of these
objects and their limits, for general information structures, priors,
and networks of counterparty relationships. Our key device is an \emph{interaction
structure} combining the network and agents' beliefs, which we analyze
using Markov methods. This device allows us to nest classical beauty
contests and network games within one model and unify their results.
Two applications illustrate the techniques: The first characterizes
when slight optimism about counterparties' average expectations leads
to \emph{contagion of optimism }and extreme asset prices. The second
describes the \emph{tyranny of the least-informed}: agents coordinating
on the \emph{prior} expectations of the one with the worst private
information, despite all having nearly common certainty, based on
precise private signals, of the ex post optimal action.
\end{abstract}

\date{September 2020}

\thanks{Golub: Department of Economics, Harvard University, Cambridge, U.S.A.,
bgolub@fas.harvard.edu, website: bengolub.net.
Morris: Department of Economics, MIT,
U.S.A., semorris@mit.edu, website: economics.mit.edu/faculty/semorris.
We are grateful for conversations with Nageeb Ali, Dirk Bergemann,
Larry Blume, Ben Brooks, Andrea Galeotti, Jason Hartline, Tibor Heumann,
Matthew O. Jackson, Bobby Kleinberg, Eric Maskin, Dov Samet, and Omer
Tamuz; for comments from Selman Erol, Alireza Tahbaz-Salehi and Muhamet
Yildiz, who served as discussants; as well as many questions and comments
from seminar and conference participants. We especially thank Ryota
Iijima for several inspiring conversations early in this project.
Cristian Gradinaru and Georgia Martin, and Eduard Talam\`{a}s provided
excellent assistance in preparing the manuscript.}
\maketitle

\section{Introduction}

Consider a situation in which each agent has strong incentives to
match the behavior of others. An outcome that agents coordinate on
in such a setting has been called a \textit{convention} in philosophy
and economics (see \citet*{Lewis1969}, \citet*{young1996economics},
and \citet*{shwi96}). In deciding how to coordinate, agents will
take into account their beliefs about (i) the state of the world,
which determines the best action; and (ii) one another's actions.
Agents may differ from one another, and hence have an incentive to
choose differently, for three reasons: first, because they are asymmetrically
informed; second, because they interpret the same information differently\textemdash that
is, they have different priors; and third, because they differ in
whom they want to coordinate with. Which conventions emerge in such
an environment will depend on the information asymmetries, the heterogeneous
prior beliefs, and the network describing the coordination motives
of agents. Our purpose is to characterize this dependence.

\medskip{}

We informally describe a simple model of this environment: a coordination
game with linear best responses. Nature draws an \emph{external state}
$\theta$, and each agent $i$ chooses a real-valued \emph{action}
$a^{i}$ based on some private information. This occurs simultaneously.
Agents' payoffs capture two motivations: First, they seek to coordinate
with a \emph{basic random variable} $y(\theta)$\textemdash a random
variable, common to everyone, that depends only on the external state;
second, they seek to take actions that are close to the actions that
others take\textemdash the various $a^{j}$ for $j\neq i$. The disutilities
they experience are proportional to the squares of the differences
between $a^{i}$ and these various targets; this feature induces best
responses linear in an agent's expectations of $y$ and others' actions.
A \emph{network} of weights captures the coordination concerns of
the agents\textemdash that is, \emph{which} others each agent cares
about coordinating with. 

If just the coordination motive were present, with no desire to match
the basic random variable, there would be a continuum of equilibria.
Indeed, for any action, there would be an equilibrium with everyone
choosing that action. The choice of action would be an \emph{arbitrary}
convention. We will be interested in the case where the convention
is not arbitrary, because agents put some weight on the accuracy motive\textemdash matching
the basic random variable\textemdash while still being strongly motivated
to choose actions close to others' actions. In this case, it turns
out there is a unique equilibrium. When the weight on others' actions
is high, it can be shown that agents essentially choose a common action.
We call this action\textemdash the common action played in the limit\textemdash the
\emph{convention}. If there were common knowledge of the external
state $\theta$, the convention would be equal to $y(\theta)$, the
value everyone seeks to match. But we are interested in characterizing
the convention when there is incomplete information about the state. 

The convention will depend on higher-order expectations of the agents.
Suppose Ann cares mainly about coordinating with Bob, who cares mainly
about coordinating with Charlie. (Recall that the agents all care
a little about matching their own expectations of the external variable.)
Then Ann's expectation of Bob's expectation of Charlie's expectation
of the external variable becomes relevant for Ann's decision. In this
scenario, each agent is seeking to coordinate with only one other,
but in the general model each seeks to match a (weighted) average
of the actions of several others. By an elaboration of the above reasoning
about Ann, Bob, and Charile, \emph{higher-order average expectations}
become relevant: Each agent cares about the average of his neighbors'
expectations of the average of \emph{their} neighbors' expectations
of the external variable, and so on. Thus our analysis of coordination
games leads naturally to a study of higher-order average expectations.
We will define the \emph{consensus} \emph{expectation} to be (essentially)
the limit of such higher-order average expectations as the order becomes
large. The consensus expectation will equal the convention that obtains
in the linear best-response game described above, in the limit as
agents' coordination concerns dominate. We will focus on this limit,
though many of the techniques we will develop can be extended to study
the case where coordination motives are not dominant.

\medskip{}

We will report three kinds of substantive results about consensus
expectations. To establish these results, we introduce a key technical
device: a Markov matrix on the union of agents' signals, which we
call an \emph{interaction structure}, capturing both the network and
agent's beliefs. A key observation is that consensus expectations
are determined by the stationary distribution corresponding to the
Markov matrix. We now present the substantive results, and we discuss
the technique in more detail at the end of the Introduction.

\subsection*{Unifying and Generalizing Network and Asymmetric Information Results.}

The first results unify and generalize facts known in the literatures
on network games and on asymmetric information: 

(a) Suppose agents have the same information but may have heterogeneous
beliefs about $\theta$\textemdash that is, different priors, which
are commonly known. Then the consensus expectation is simply a weighted
average of agents' heterogeneous prior expectations of the external
random variable. The weight on an agent's expectation is his \emph{eigenvector
centrality} in the network. This corresponds to the seminal result
of \citet*{Ballester2006} on equilibrium actions in certain network
games being weighted averages of individuals' ideal points, with someone's
weight determined by the extent to which others want to directly and
indirectly coordinate with him. The appearance of network centrality
here\textemdash a statistic of individuals defined from the matrix
of coordination weights\textemdash is a consequence of the matrix
algebra that naturally appears when studying higher-order average
expectations.\footnote{For recent surveys of economic applications related to network centrality,
see \citet*[Section 2.2.4]{JacksonBook}, \citet*{acemoglu-shocks},
\citet*{zenou-games}, and \citet*{golub-sadler}.}

(b) If there is \emph{asymmetric} information, but agents have common
prior beliefs, then the consensus expectation is equal to the (common
prior) ex ante expectation of the external state. Thus consensus expectations
are independent of the network structure, and also independent of
all features of the information structure except the common prior.
This result turns out to be a corollary of the result of \citet*{samet1998iteratedA}. 

(c) Embedding both (a) and (b), if agents have both heterogeneous
prior beliefs and asymmetric information but a \emph{common prior
on signals},\footnote{That is, a common prior on how the signal random variables are distributed.}
then the consensus expectation is equal to a weighted average of agents'
different ex ante expectations of the basic random variable. Just
as in (a), the weight on an agent's expectation is his \emph{eigenvector
centrality} in the network.\footnote{This decomposition separates, in a suitable sense, the effects of
the network and the beliefs. A companion paper, \citet*{GolubMorris2016},
gives the necessary and sufficient condition on the information structure
for this sort of weighted average decomposition to be possible.} This goes beyond existing work on network games with incomplete information
due to \citet*{cadp15}, \citet*{Marti2015}, \citet*{behm15} and
\citet*{bbdj15}: we will discuss these connections in Section \ref{sec:cpa-signals}
when we have introduced the model and key results. 

\subsection*{Contagion of Optimism.}

Our second category of results studies \emph{second-order optimism}.\emph{
}We assume that each agent, given any signal, assesses his average
counterparty as more optimistic than himself about the value of the
basic random variable, unless the agent himself has a first-order
expectation that is already very high (close to the highest induced
by any signal). Agents whose expectations are high may be somewhat
pessimistic: they may assess the average counterparty as less optimistic
than themselves. 

We study when arbitrarily slight second-order optimism leads consensus
expectations to be very high\textemdash near highest possible expectation
of $y$\textemdash via a contagion of optimism through higher-order
expectations. The proof is via a reduction to a Markov chain inequality.
The key subtlety in the analysis is: how much pessimism can be allowed
without destroying the contagion of optimism? We give a bound that
answers this question, and describe a sense in which this bound is
tight (Section \ref{subsec:optimism-tightness}). Recent work of \citet*{han16}
discusses a different contagion of optimism in a CARA-normal rational
expectations model. We examine connections with related models in
Section \ref{subsec:optimism-related}. 

\subsection*{Tyranny of the Least-Informed.}

Third, we consider a setting where agents start with heterogeneous
priors about the external state but share a \emph{common interpretation
of signals}.\footnote{An \emph{interpretation} of a signal random variable is its conditional
distribution given the state, in line with the terminology of \citet*{Kandel}
and \citet*{Acemoglu2016}. } That is, agents observe signals of the external state. They agree
on the probability of any particular signal of a given agent conditional
on any external state.\footnote{In the environment we study for this application, the signals are
conditionally independent given the state, so that signals are correlated
only through the state.} However, their priors over external states may differ, and thus their
interim beliefs may not be compatible with a common prior. Given common
interpretation of signals, it makes sense to define notions of more
and less precisely informed agents, because the distributions of signals
given the external state (and thus the levels of noise in them) are
common knowledge. 

We show that, in a suitable sense, the consensus expectation approximates
the ex ante expectation of the agent whose private information is
\emph{least} precise. This is true even if all agents have very precise
private signals about the state, as long as the least-informed has
signals sufficiently less precise than others. The quantitative details
of how to define ``sufficiently'' are subtle, and rely on a Markov
chain connection that we discuss next.

\subsection*{The Interaction Structure and Markov Formalism.}

The techniques underlying the results discussed above are based on
a Markov matrix description of higher-order average expectations.
While we defer most of the details until Section \ref{sec:The-Interaction-Structure},
when we have more notation, the basic idea is simple. We define a
Markov process whose state space is the \emph{union} of all agents'
signals. Transition probabilities between any two states combine both
the network weights and the subjective probabilities of the agents.
In particular, the transition probability from a signal $t^{i}$ of
agent $i$ to a signal $t^{j}$ of agent $j$ is defined as the product
of (i) the network weight that $i$ places on $j$ and (ii) the subjective
probability that agent $i$, given signal $t^{i}$, places on $t^{j}$.
We call the transition matrix of this Markov process the \emph{interaction
structure}, and it is our key technical device. This formalism treats
beliefs and network weights entirely symmetrically. This symmetric
treatment enables the analysis to be reduced to Markov chain results,
which provide both a tool and novel insights.\footnote{The symmetric treatment follows \citet*{morris1997interaction}. As
discussed in detail in Section \ref{subsec: samet}, our approach
echoes \citet*{samet1998iteratedA} in using a Markov process to represent
incomplete information, although our Markov process is actually a
different one in significant ways.} Other work on network games with incomplete information\textemdash \citet*{cadp15},
\citet*{Marti2015}, \citet*{behm15} and \citet*{bbdj15}\textemdash does
not use this general device and must develop more tailored techniques.

The essence of our approach is that the iteration of the Markov matrix
associated with the interaction structure enables a brief, explicit
description of higher-order average expectations: The $n$$^{\text{th}}$-order
average expectations can be obtained by suitably combining the $n$-step
transition probabilities of the Markov process with the first-order
expectations associated to various signals.

To study consensus expectations, we consider the limit as $n$ grows
large. Under suitable conditions, in this limit the Markov transition
probability to any state\textemdash regardless of where the process
starts\textemdash becomes the stationary probability of that state.
This can be used to show that \emph{the stationary distribution of
the Markov process determines the consensus expectation}. Indeed,
the consensus expectation turns out to be a weighted average of first-order
expectations given various signals $t^{i}$. The weight on a signal
$t^{i}$ is its weight in the stationary distribution of the Markov
process.

Thus, our results on the consensus expectation are proved by studying
the stationary distribution of the Markov process and deriving properties
of it from more primitive assumptions about the environment. For example,
in the analysis of the contagion of optimism, the essential idea is
that when second-order optimism holds, probability mass in the Markov
process flows on average to signals associated with higher first-order
expectations of the basic random variable. It can be shown that, as
a consequence, states with high first-order expectations have a larger
share of the stationary probability. By our description of the consensus
expectation as a weighted average of first-order expectations, with
weights given by the stationary probabilities, it follows that the
consensus expectation is high.

Other results rely on different reasoning. The most technically involved
arguments are the ones associated with the tyranny of the least-informed.
These arguments rely on perturbation bounds for Markov chains, which
are used to show that the priors of highly informed agents cannot
play a substantial role in the stationary distribution that determines
the consensus expectation. Overall, our main methodological claim\textemdash illustrated
by the various applications\textemdash is that the structure of higher-order
expectations is illuminated by the Markov formalism. 

\medskip{}

The remainder of the paper is organized as follows. Section \ref{sec:model}
presents the environment, defines higher-order average expectations
and consensus expectations, and illustrates them with some simple
examples. Section \ref{sec:why-matter} motivates higher-order average
expectations and consensus expectations by discussing a coordination
game and an asset market where they are relevant. Section \ref{sec:The-Interaction-Structure}
presents our key technical device, the interaction structure, and
the correspondence between higher-order average expectations and statistics
of a Markov process. Section \ref{sec:consensus-and-network} relates
the interaction structure to the underlying network. Section \ref{sec:cpa-signals}
relates consensus expectations to agents' priors. Together these results
unify and extend the known network games and incomplete-information
results. Section \ref{sec:optimism} focuses on higher-order optimism,
while Section \ref{sec:Ignorance} reports our results on the tyranny
of the least-informed. Section \ref{sec:discussion} is a discussion
of relations to the literature, subtleties, and extensions.

\section{Model\label{sec:model}}

\subsection{The Information Structure}

\subsubsection{States, Signals, and Expectations}

\label{sec:info-structure}There is a finite set $\Theta$ of \emph{states
of the world}. There is a finite set $N$ of agents. Associated to
each agent $i\in N$ is a finite set $T^{i}$ of \emph{signals} (i.e.,
possible signal realizations). and these sets of signals are disjoint
across agents. Let $T=\prod_{i\in N}T^{i}$ be the product of all
the signal spaces, with a typical element being a tuple $t=(t^{i})_{i\in N}$;
let $T^{-i}=\prod_{j\in N\setminus\left\{ i\right\} }T^{j}$ be the
product of the signal spaces of all the others, viewed from $i$'s
perspective. An agent's signal fully determines all the information
he has, including the information he has about others' signals. Let
$\Omega=\Theta\times T$ be the set of all \emph{realizations}.

For each $i$ and signal $t^{i}$, there is a \emph{belief} $\pi^{i}(\cdot\mid t^{i})\in\Delta(\Theta\times T^{-i})$\textemdash that
is, a probability distribution over $\Theta\times T^{-i}$. This is
the \emph{interim }or \emph{conditional} belief that agent $i$ has
when he gets signal $t^{i}$. We introduce some notation to refer
to marginal distributions: $\pi^{i}(t^{j}\mid t^{i})$ denotes the
probability this belief assigns to agent $j$'s signal being $t^{j}$.\footnote{That is, $\pi^{i}(\{(\hat{\theta},\hat{t}^{-i}):\widehat{\theta}\in\Theta,\hat{t}^{j}=t^{j}\})$.}
For states $\theta\in\Theta$, the notation $\pi^{i}(\theta\mid t^{i})$
has an analogous definition. We refer to $\bm{\pi}=(\pi^{i}(\cdot\mid t^{i}))_{i\in N,\:t^{i}\in T^{i}}$
as the \emph{information structure}.\emph{}\footnote{As always, uncertainty about how signals are generated can be built
into this description of an information structure. Thus, following
\citet*{Harsanyi1968}, the information structure itself is taken
to be common knowledge. For more on this see \citet*[p. 1237]{auma76}
and \citet*{Brandenburger1993}.} In situations where only interim beliefs matter, we will use the
language of \emph{types}. That is, we will identify each signal with
a corresponding (belief) type of the agent. If signal $t^{i}$ induces
a certain belief over $\Theta\times T^{-i}$, we will say that type
$t^{i}$ (of agent $i$) has that belief. We will call $T$ the \emph{type
space}. On the other hand, when we wish to emphasize the ex ante stage
and the literal process of drawing signals, we will use the language
of signals.

A random variable measurable with respect to $i$'s information is
a function $x^{i}\colon T^{i}\to\mathbb{R}$, i.e., an element of
$\mathbb{R}^{T^{i}}$ (this set being defined as the set of functions
from signals in $T^{i}$ to real numbers). Given a random variable
$z:\Omega\to\mathbb{R}$, let $E^{i}z\in\mathbb{R}^{T^{i}}$ give
$i$'s conditional expectation of $z$. It is defined by 
\begin{equation}
(E^{i}z)(t^{i})=\sum_{(\theta,t^{-i})\in\Theta\times T^{-i}}\pi^{i}(\theta,t^{-i}\mid t^{i})\:z(\theta,t^{i},t^{-i}).\label{def:conditional-ex}
\end{equation}
The summation runs over all $(\theta,t^{-i})$, and states are weighted
using the probabilities assigned by the interim belief $\pi^{i}(\cdot\mid t^{i})$.
We will often abuse notation, as we have done here, by dropping parentheses
in referring to elements of $\Omega$ in the arguments of beliefs
and random variables.

\subsubsection{Priors\label{subsec:Priors}}

The information structure was defined above in terms of agents' interim
beliefs, i.e., their beliefs about external states and others' signals
conditional on their own signals. This interim information is enough
to define higher-order average expectations and to state our main
results. However, we are interested in the ex ante interpretation
of our results: There is a prior stage before agents observe their
own signals, and thus where they face uncertainty as to what signals
they will observe. 

We write $(\mu^{i})_{i\in N}$ for agents' ex ante beliefs, with $\mu^{i}\in\Delta(T^{i})$.
Combined with conditional beliefs $\pi^{i}(\cdot\mid t^{i})\in\Delta(\Theta\times T^{-i})$,
there is a prior $\mathbf{P}^{\mu^{i}}\in\Delta(\Omega)$ on the entire
space of realizations, assigning to any $(\theta,t)\in\Omega$ a probability\footnote{Note that the probability under $\mathbf{P}^{\mu^{i}}$ of any subset
of $\Omega$ can be written as a sum of probabilities defined in equation
(\ref{eq:def-prior}), and a similar statement holds for the interim
probabilities $\pi^{i}(\cdot\mid t^{i})$.}
\begin{equation}
\mathbf{P}^{\mu^{i}}\left(\theta,t\right)=\sum_{t^{i}\in T^{i}}\mu^{i}(t^{i})\:\pi^{i}(\theta,t^{-i}\mid t^{i}).\label{eq:def-prior}
\end{equation}
If one started from agent $i$'s prior $\mathbf{P}^{\mu^{i}}\in\Delta(\Omega)$,
one would define conditional beliefs $\pi^{i}(\cdot\mid t^{i})\in\Delta(\Theta\times T^{-i})$
by updating according to Bayes' rule. 

The probability measure $\mathbf{P}^{\mu^{i}}$ gives rise to an ex
ante expectation operator, 
\[
\mathbf{E}^{\mu^{i}}z=\sum_{\omega\in\Omega}\mathbf{P}^{\mu^{i}}(\omega)\:z(\omega)=\sum_{t^{i}\in T^{i}}\mu^{i}(t^{i})\:E^{i}z.
\]
To emphasize when an ex ante perspective is being taken, we adopt
the convention that ex ante probabilities, expectations, etc. are
in bold.

We will later be interested in what an agent's ex ante beliefs would
be if we had fixed his conditional beliefs $\pi^{i}(\cdot\mid t^{i})\in\Delta(\Theta\times T^{-i})$
but endowed him with alternative prior beliefs. Priors for $i$ other
than the true priors $\mu^{i}$ are denoted by $\lambda^{i}\in\Delta(T^{i})$,
and we use $\lambda^{i}$ in place of $\mu^{i}$ in the notations
introduced above.

\subsection{The Network\label{sec:network}}

For each pair of agents, $i$ and $j\text{,}$ there is a number $\gamma^{ij}\in[0,1]$,
where $\sum_{j\in N}\gamma^{ij}=1$, with the interpretation that
agent $i$ assigns ``weight'' $\gamma^{ij}$ to agent $j$. A matrix
$\Gamma$, whose rows and columns are indexed by $N$ and whose entries
are $\gamma^{ij}$, records these weights and is called the \emph{network}.
The fact that the weights of any agent add up to $1$ corresponds
to this matrix being row-stochastic.

The network is to be contrasted with the \emph{information structure}
encoded in the interim beliefs $\pi^{i}(\cdot\mid t^{i})$. One interpretation
of the network weight $\gamma^{ij}$, which will be used when we discuss
coordination games, is that it measures how much agent $i$ cares
about the action of $j$. We define $N_{i}$, the \emph{neighborhood}
of \emph{i}, to be the set of $j$ such that $\gamma^{ij}>0$, and
the elements of $N_{i}$ are $i$'s \emph{neighbors.} Note that $j$
may be a neighbor of $i$ without $i$ being a neighbor of $j$.

We now define an important set of statistics arising from the network.
\begin{defn}
\label{def:eigenvector-centrality}The \emph{eigenvector centrality
weights }of the agents are the entries of the unique row vector $e\in\Delta(N)$
satisfying $e\Gamma=e$\textemdash i.e., for each $i$,
\[
e^{i}=\sum_{j\neq i}e^{j}\gamma^{ji}.
\]
\medskip{}
Assuming that $\Gamma$ is irreducible, the Perron\textendash Frobenius
Theorem states that the eigenvector centrality weights are well-defined\textemdash that
there is indeed a unique such vector\emph{ $e$}. Moreover, the theorem
says that all the eigenvector centrality weights are positive.
\end{defn}

\subsection{Higher-Order Average Expectations}

\label{sec:iterated-average-expectations}

We now define higher-order average expectations. A \emph{basic} random
variable is a random variable measurable with respect to the external
states, i.e., a function $y:\Theta\to\mathbb{R}$, or an element of
$\mathbb{R}^{\Theta}$. Consider a random variable $y\in\mathbb{\mathbb{R}}^{\Theta}$
and define\footnote{Here, abusing notation, we have identified $y\in\mathbb{R}^{\Theta}$
in the obvious way with a random variable $z\in\mathbb{R}{}^{\Theta\times T}$,
namely, with the random variable $z$ for which $z(\theta,t)=y(\theta)$
for each $(\theta,t)\in\Theta\times T$. Equation (\ref{eqn:def-iteration})
relies on a similar understanding.} 
\begin{equation}
x^{i}(1;y)=E^{i}y\label{eqn:def-first-step}
\end{equation}
for every $i\in N$. This is $i$'s \emph{first-order expectation},
given $i$'s own signal, of $y$.

We can now define the key objects we will focus on: the iterated expectations,
or \emph{higher-order average expectations}. For $n\geq2$, given
$(x^{i}(n))_{i\in N}$, define 
\begin{equation}
x^{i}(n+1;y,\Gamma)=\sum_{j\in N}\gamma^{ij}E^{i}x^{j}(n;y,\Gamma).\label{eqn:def-iteration}
\end{equation}
This is $i$'s subjective expectation of the average of the random
variables corresponding to the previous iteration of the process;
the average is taken with respect to the network weights. 

When we do not wish to emphasize the dependence on $y$ and $\Gamma$,
or when they are clear from context, we omit these arguments.

Note that equation (\ref{eqn:def-iteration}), despite the presence
of an iteration, is defined in a static environment: Higher-order
average expectations do not correspond to dynamic updating over time,
but rather to a hierarchy of beliefs when agents are simultaneously
given different information. For this reason, these will figure in
the solution of a static game (see Section \ref{subsec:game}, and
a contrast with dynamics in Section \ref{de groot-1}). 

\subsection{Examples}
\begin{example}
If we have $\gamma^{ij}=1/|N|$ for all $i,j$, then every agent is
weighting all others equally. Such averages will turn out to be relevant
for beauty contests with homogeneous weights: $x^{i}(n)$ is a random
agent's expectation of a random agent's expectation \ldots{} of a random
agent's expectation of $y$. 
\end{example}
\begin{example}
\label{exa:cycle}Suppose the only nonzero entries of $\Gamma$ are
$\gamma^{i,i+1}=1$, where indices are interpreted modulo $|N|$,
the number of agents. This corresponds to agents being arranged in
a cycle, with each paying attention to the one with the next index.
Take, for example, $|N|=3$ (see Figure \ref{fig:cycle}). Then 
\[
x^{1}(3)=E^{1}E^{2}E^{3}y.
\]
We could continue this process, and then we would essentially look
at $(E^{3}E^{1}E^{2})^{a}E^{3}y$, where $a$ is some positive integer
(possibly with $E^{2}$ or $E^{1}E^{2}$ appended to the front). Our
study of higher-order average expectations will allow us to study
the limiting properties of this sequence.
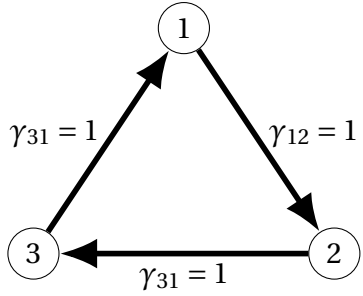
\begin{figure}
\begin{centering}
\begin{tikzpicture}[scale=1]    \coordinate [label=left:{$\gamma_{31}=1$}] (zero) at (2,1.6); \coordinate [label=right:{$\gamma_{12}=1$}] (zero) at (4,1.6); \coordinate [label=below:{$\gamma_{31}=1$}] (zero) at (3,0); \node[circle, draw] (1) at (3,3) {$1$};  \node[circle, draw] (3) at (1,0) {$3$};  \node[circle, draw] (2) at (5,0) {$2$}; \draw[line width=.75mm,-{Latex[length=5mm, width=4mm]}] (1) to (2); \draw[line width=.75mm,-{Latex[length=5mm, width=4mm]}] (2) to (3); \draw[line width=.75mm,-{Latex[length=5mm, width=4mm]}] (3) to (1);           \end{tikzpicture}
\par\end{centering}
\caption{The network of Example \ref{exa:cycle}.}
\label{fig:cycle}
\end{figure}
\end{example}

\subsection{Joint Connectedness: A Maintained Technical Assumption\label{sec:joint-connected}}

A key technical assumption\textemdash \emph{joint connectedness} of
the information structure and network\textemdash will be convenient
in formulating statements about limits of higher-order average expectations.
This assumption will be maintained unless we state otherwise.

Say that a signal $t^{j}$ (of an agent $j$) is a neighbor of a signal
$t^{i}$ (of agent $i$) if agent $j$ is a neighbor of $i$ (i.e.,
$\gamma^{ij}>0$) and agent $i$, when he observes signal $t^{i}$,
considers signal $t^{j}$ possible (i.e., $\pi^{i}(t^{j}\mid t^{i})>0$).
This defines a binary relation on the set of everybody's signals,
$S=\bigcup_{i\in N}T^{i}$. We say the information structure and network
are \emph{jointly connected} if every nonempty, proper subset $S'\subsetneq S$
contains some signal that is a neighbor of a signal not in $S'$.
We will discuss the content and significance of this assumption below
in Section \ref{sec:irred}.

\subsection{Consensus Expectations: Definition and Existence\label{subsec:consensus-expectation}}

An object central to our general theoretical results and the applications
will be a kind of limit of higher-order average expectations as we
consider many iterations. 
\begin{defn}
\label{def:consensus-expectation}For any information structure $\bm{\bm{\pi}}$,
network $\Gamma$, and basic random variable $y$, the \emph{consensus
expectation} $c(y;\bm{\pi},\bm{\Gamma})$ is defined to be any entry
of the vector
\begin{equation}
\lim_{\beta\uparrow1}\left(1-\beta\right)\left({\displaystyle \sum\limits _{n=1}^{\infty}\beta^{n-1}x^{i}(n;y)}\right),\label{eq:c-limit-introduced}
\end{equation}
for any $i$, if the limit exists (in the sense of pointwise convergence)
and is equal to a constant vector. 
\end{defn}
The vector in (\ref{eq:c-limit-introduced}) is sometimes called an
\emph{Abel average} of the sequence $\left(x^{i}(n;y)\right){}_{n=1}^{\infty}$
\citep*[see, e.g.,][]{Kozitsky2013}. Proposition \ref{prop:char}
in Section \ref{sec:The-Interaction-Structure} below asserts that
the consensus expectation is well-defined under the maintained assumption
of joint connectedness. 

The consensus expectation is equal to any entry of the simple limit
$\lim_{n\to\infty}x^{i}(n;y)$ if the latter exists. It also coincides
with the Cesàro limit, which is obtained by taking simple averages
over many values of $n$. We will discuss these issues further in
Section \ref{sec:consensus-expectation}.

\section{Why Higher-Order Average Expectations and Consensus Expectations
Matter}

\label{sec:why-matter}

We now discuss two economic problems where higher-order average expectations
arise. First, we consider the network game with incomplete information
discussed in the Introduction, where equilibrium actions are weighted
averages of higher-order average expectations. Second, we describe
a stylized asset market with fragmented markets, where asset prices
reduce to the solution of the game, and are thus also weighted averages
of higher-order average expectations. In each of these two cases,
we will (i) show how outcomes are characterized by higher-order average
expectations; (ii) motivate the study of consensus expectations\textemdash a
limit of higher-order average expectations; and (iii) interpret our
later results in the context of these applications. 

\subsection{Coordination\label{lbr}}

How will a group of agents coordinate their behavior when they have
strong incentives to take the same action as others but have different
beliefs about what the best action to take is? We consider a class
of games with linear best responses where each agent wants to set
her action equal to a weighted average of (i) her expectation of a
random variable and (ii) the weighted average of actions taken by
others. We show how the equilibrium is determined by higher-order
average expectations and then focus on the limit as coordination concerns
dominate. There will be a particular single action taken in this limit
by all agents after all signals\textemdash ``the convention.'' We
first describe the game. 

\subsubsection{The Game\label{subsec:game}}

We will consider an incomplete-information game where payoffs depend
on the states of the world, $\Theta$. Beliefs and higher-order beliefs
about $\Theta$ are described by the belief functions introduced in
Section \ref{sec:info-structure}.\footnote{If we identify types with their (interim) beliefs, we can say that
the beliefs are encoded in the type space.} The strategic dependencies are encoded in a network $\Gamma$. We
also assume that $\gamma^{ii}=0$ for all $i$.\footnote{The assumption that the diagonal is $0$ is the most natural one for
this game. Analogous results hold without this assumption and have
game-theoretic interpretations. See Section \ref{subsec:Heterogeneous}.} 

The game will also depend on $y$, a basic (i.e., $\theta$-measurable)
random variable with support in the interval $\left[0,M\right]$.\footnote{We will focus on the case where agents care about the same basic random
variable. \ But the analysis extends readily to the case where agents
care about different random variables, since their heterogeneous expectations
of random variables conditional on their signals can be interpreted
as agent-specific random variables. \ See Section \ref{heterogeneous-1}
for further discussion. \ } 

We will consider the ``$\beta$-game\textquotedbl{} parameterized
by $\beta\in\left[0,1\right]$. Each agent $i$ chooses an action
$a^{i}\in\left[0,M\right]$, and the best-response action of agent
$i$ after observing signal $t^{i}$ is given by 
\[
a^{i}=(1-\beta)E^{i}y+\beta{\displaystyle \sum\limits _{j\neq i}\gamma^{ij}E^{i}a^{j}},
\]
where other players' actions are viewed as random variables that depend
on their own signal realizations. The best response can be derived
from a quadratic loss function, where the ex post utility of agent
$i$ under realized state $\theta\in\Theta$, if action profile $a=\left(a^{i}\right)_{i\in N}\in\mathbb{R}^{N}$
is played, is 
\[
u^{i}\left(a^{i},\theta\right)=-\left(1-\beta\right)\left(a^{i}-y(\theta)\right)^{2}-\beta{\displaystyle \sum\limits _{j\neq i}\gamma^{ij}\left(a^{i}-a^{j}\right)^{2}\text{.}}
\]

A ``meetings'' interpretation of the weights $\gamma^{ij}$ is that
$i$ has to commit to an action before knowing which agent he will
interact with, and $i$ assesses that the probability of interacting
with $j$ is $\gamma^{ij}$. 

\subsubsection{Solution of the Game for Any $\beta<1$\label{subsec:solution-game}}

To summarize the previous section, the \emph{environment} in which
the game is played is described by a tuple consisting of an external
random variable, a network, and a coordination weight: $(y,\Gamma,\beta)$.
A strategy of agent $i$ in the incomplete-information game, $s^{i}:T^{i}\to\mathbb{R}$,
specifies an action for each signal. Write $s^{i}(t^{i})$ for the
action chosen by agent $i$ upon observing signal $t^{i}$. Then agent
$i$'s best response to strategy profile $s=\left(s^{i}\right)_{i\in N}$
is given by 
\begin{equation}
\text{BR}^{i}(s)=(1-\beta)E^{i}y+\beta{\displaystyle \sum\limits _{j\neq i}\gamma^{ij}E^{i}s^{j}}.\label{eq:best-response}
\end{equation}

To establish (\ref{rationalizable}), write $R^{i}(k)$ for the set
of $i$'s pure strategies surviving $k$ rounds of iterated deletion
of strictly dominated strategies.\footnote{One reason to focus on rationalizability is that because we do not
have a common prior, there is some inconsistency in using a solution
concept (equilibrium) which builds in common prior beliefs about strategic
behavior \citep*[see][]{dekel2004learning}. \label{fn:rationalizability-versus-equilibrium}} The map $\text{BR}(s):[0,M]^{S}\to[0,M]^{S}$ is a contraction mapping
(with Lipschitz constant $\beta$). Thus, the sets $R^{i}(k)$, which
are produced by the repeated application of this map to the set $[0,M]$,
must converge to a single point satisfying $s=\text{BR}(s)$, which
is an equilibrium of our game. A more detailed proof can be found
in an Appendix, Section \ref{subsec:proof-of-rationalizable}.

This analysis is the asymmetric version of the analysis in \citet*{mosh02}. 

\subsubsection{Conventions: Equilibrium for $\beta=1$ and $\beta\uparrow1$}
\begin{fact}
\label{fact:rationalizable}If $\beta<1$, the $\beta$-game in the
environment given by $(\Gamma,\beta,y)$ has a unique rationalizable
strategy profile, and it is given by\footnote{This game has a unique equilibrium (even with unbounded action spaces).
\ This follows from the observation that the game is best-response
equivalent to a team decision problem, and uniqueness in the team
decision problem is shown in \citet*{radner1962team}. \citet*{ui09}
gives a general statement of this result, expressed in the language
of Bayesian potential functions. Since this is a game with strategic
complementarities, bounded action spaces imply that the unique equilibrium
is the unique action strategy profile surviving iterated deletion
of strictly dominated strategies (\citet*{miro90}). Our proof of
Fact \ref{fact:rationalizable} is established by explicitly calculating
the iterated elimination of dominated strategies.} 
\begin{equation}
s_{*}^{i}\left(y,\Gamma,\beta\right)=\left(1-\beta\right)\left({\displaystyle \sum\limits _{n=1}^{\infty}\beta^{n-1}x^{i}(n;y,\Gamma)}\right)\text{.}\label{rationalizable}
\end{equation}
\end{fact}
There is a sharp distinction between the game with $\beta<1$ and
the game with $\beta=1$. In the latter case, there is a continuum
of equilibria, one for each $a\in[0,M]$. In these equilibria, agents
all choose the same action independent of their signals and thus of
the state. To see why, recall that every agent's action must be equal
to his (weighted) expectation of others' actions. But now consider
the highest action ever played in some equilibrium (i.e., given some
signal of some agent). The agent $i$ taking that highest action at
some signal $t^{i}$ must be sure that that highest action is being
taken by every other agent $j$ who observes a signal $t^{j}$ that
$i$ considers possible when he observes signal $t^{i}$. Now, however,
the same logic applies to agent $j$ observing that signal $t^{j}$.
Continuing in this way, our joint connectedness assumption implies
that the highest action must be played by all agents for all signal
realizations. This argument and result appear in \citet*{shwi96},
who label the resulting play\textemdash constant across agents and
signals\textemdash a convention, because each agent is always choosing
the same action and is choosing that action because others do. 

To summarize: When $\beta<1$, there is a unique equilibrium, with
agents' actions depending on their higher-order expectations of $y$.
When $\beta=1$, there is a continuum of ``conventional'' equilibria\@.
What happens as $\beta\uparrow1$? The play is described by a limit
of unique equilibria, which turns out to be well-defined: 
\[
\lim_{\beta\uparrow1}\left(1-\beta\right)\left({\displaystyle \sum\limits _{n=1}^{\infty}\beta^{n-1}x^{i}(n)}\right).
\]
By an application of the argument of the previous paragraph to the
limiting payoffs, under joint connectedness the limit must feature
``conventional'' play, not depending on one's signal or identity.
The existence of the limit and a characterization of the action played
in it will be formalized in the next section; the main result is Proposition
\ref{prop:char}. The limit can be seen as a selection among the continuum
of equilibria of the $\beta=1$ game. It is telling us how conventional
play is determined when there is an arbitrarily small amount of dependence
of the payoffs on some basic random variable. The basic random variable
can be interpreted as a ``cue'' that orients players' coordination. 

Note that the statements made about the game before we started considering
the $\beta\uparrow1$ limit hold for any $\beta\in[0,1]$. From now
on, we will focus on the $\beta\uparrow1$ limit, motivated by the
interpretation of it just given, as a refinement of the coordination
game (as well as a parallel motivation we are about to present, based
on frequent trade of an asset).\footnote{\citet*{Weinstein2007} have argued that, in a fixed linear best response
game, very high-order beliefs have only a small impact on rationalizable
play; this constrasts with the better-known observation in \citet*{weinstein2007structure}
that very high-order beliefs can have an arbitrarily high impact in
general games. We get sensitivity to high level higher-order beliefs
in linear best response games because are looking at $\beta\uparrow1$
limit and thus a sequence of different games. Subtleties of such comparisons
are also discussed in \citet*{typicaltypes}.}

Our main results focus on $\beta$ begin very close to, but not equal
to, $1$. Some of our results apply to, or have implications for,
the more general situation with $\beta$ much smaller than $1$, and
we discuss that case when appropriate. 

\subsubsection{Conventions with High Coordination Weights: Preview of Main Results\label{subsec:game-applications-of-main-results}}

Our results in Sections \ref{sec:cpa-signals}, \ref{sec:optimism},
and \ref{sec:Ignorance} characterize that limit convention in some
environments: 
\begin{enumerate}
\item Under the common prior assumption, the convention is equal to the
common ex ante expectation of $y$. If agents share a common prior
on signals, but not necessarily on states, then the convention is
equal to a weighted average of (different) ex ante expectations that
the agents hold of $y$, with each agent's expectation weighted by
his eigenvector centrality in the interaction network $\Gamma$.
\item If all agents always have a small amount of second-order optimism
(believing that their average counterparty is a bit more optimistic
than they are), the convention will equal the highest interim expectation
ever held by any agent.
\item If there is common interpretation of signals and one agent is sufficiently
less informed than all other agents, then the convention will equal
the ex ante expectation of that least-informed agent.
\end{enumerate}

\subsection{Asset Pricing\label{asset pricing}}

\citet*[p. 156]{keynes19731936} famously likened investment to a
``beauty contest'' whose outcome depends on higher-order beliefs: 
\begin{quotation}
\ldots\ professional investment may be likened to those newspaper
competitions in which the competitors have to pick out the six prettiest
faces from a hundred photographs, the prize being awarded to the competitor
whose choice most nearly corresponds to the average preferences of
the competitors as a whole; so that each competitor has to pick, not
those faces which he himself finds prettiest, but those which he thinks
likeliest to catch the fancy of the other competitors, all of whom
are looking at the problem from the same point of view. It is not
a case of choosing those which, to the best of one's judgement, are
really the prettiest, nor even those which average opinion genuinely
thinks the prettiest. We have reached the third degree where we devote
our intelligences to anticipating what average opinion expects the
average opinion to be. And there are some, I believe, who practise
the fourth, fifth and higher degrees.\textquotedblright{} 
\end{quotation}
Keynes is presumably not suggesting that the newspaper competition
winner is completely independent of ``prettiness'' but rather that
each competitor has an incentive to try to match the average expectation
of prettiness, and then some average expectation of such average expectations,
and so on. We will study an asset pricing model where asset prices
will correspond to solutions of the coordination game above and thus
to the description of investment behavior that Keynes gives. 

\subsubsection{Asset Market\label{subsec:asset-market-description}}

Suppose that there are several populations or classes, indexed by
the elements of $N$, and each of these consists of a continuum of
infinitesimal \emph{traders}. \ There is an asset whose payoff will
depend on the realization of a random variable $y$ that is measurable
with respect to $\Theta$ and that takes values in $[0,M]$. The beliefs
and higher-order beliefs of traders in class $i$ about the state
space $\Theta$ will be given by the belief function $\pi^{i}$ defined
in the general model; in particular, they share the same belief function.
Each trader in class $i$ will also observe the same signal $t^{i}$.
Thus, they share the same interim beliefs. All traders are risk-neutral
and there is no discounting. A single unit of the asset will be traded
among all classes of traders. There is a network $\Gamma$, which
will determine where traders resell their assets in a way we are about
to describe.

The trading game works as follows. Time is discrete. At each time
$t$, one trader (say, in class $i$) enters owning the asset. With
probability $\beta$, the state is realized and the owner of the asset
consumes the realization of the asset (with the interpretation that
this corresponds to liquidity needs). \ He then exits the game. \ If
not (and so with probability $1-\beta$), a class of traders $j$
is selected randomly (and exogenously).\ The asset owner believes
that class $j$ is selected with probability $\gamma^{ij}$. \ The
owner must then sell the asset in a market consisting of all traders
of class $j$ who have not yet exited. \ There is Bertrand competition
in market $j$, with each buyer (i.e., remaining trader in class $j$)
offering a price $p$ and the seller (in class $i$) deciding to whom
to sell the asset. \ We then enter period $t+1$ with the chosen
buyer in class $j$ holding the asset. 

\subsubsection{Equilibrium Asset Prices}

We will consider symmetric Markov subgame perfect equilibria of the
asset trading game described in Section \ref{subsec:asset-market-description}.
By ``symmetric Markov,'' we mean that each trader's offer will depend
only on the class to which he belongs and the class of the current
owner from whom he is buying. 

The main result about this asset market is that there is a unique
symmetric, Markov, subgame-perfect equilibrium where, whenever the
asset is sold in market $j$, the traders in market $j$ with signal
$t^{j}$ offer a price equal to $s(t^{j})$, where $s=s^{\ast}\left(\beta\right)$
as defined in (\ref{rationalizable}), and owners sell to any trader
in class $j$ offering the highest price. \ In other words, traders
always set prices equal to the equilibrium of the linear best-response
game of the previous section. \ To see why, note first that a trader's
willingness to pay for the asset does not depend on whom he is buying
the asset from. \ Also, observe that in a symmetric equilibrium,
traders must be setting prices equal to their willingness to pay.
\ Thus equilibrium asset prices must satisfy equation (\ref{rationalizable})
i.e., the equilibrium condition from the linear best-response game.\ \  

Our analysis does depend on the restriction to symmetric Markov strategies
and equilibrium rather than rationalizability as a solution concept.\footnote{Recall footnote \ref{fn:rationalizability-versus-equilibrium} on
the comparison of rationalizability and equilibrium in this context.} If we did not impose the Markov assumption, there would be ``bubble''
equilibria, with the asset price growing exponentially. We also used
the assumption of equilibrium in our analysis, when we directly assumed
that the prices satisfy the equation (\ref{rationalizable}), rather
than (as we did in Section \ref{subsec:solution-game}) arguing that
this condition follows from some weaker solution concept. We used
symmetry when we assumed that all members of a given class price the
asset the same way whenever they have the opportunity to buy. 

\subsubsection{Asset Pricing with High-Frequency Trading: Preview of Main Results\label{subsec:Motivation}}

Taking the limit $\beta\uparrow1$ now corresponds to requiring faster
and faster trade while holding time preferences fixed. The high-frequency
limit price will be the same in every market for every signal. This
price turns out to be the consensus expectation.\footnote{\citet*{steiner2014price} have used high-frequency limits to show
a similar convergence to public random variables.} Our main results below have implications for the asset prices which
parallel the statements of \ref{subsec:game-applications-of-main-results}
applied to the game. 

\subsubsection{Techniques and Related Models of Asset Pricing}

In the special case where the network is uniform, we can could have
derived the same asset pricing formula in a standard dynamic CARA-normal
rational expectations model, with overlapping generations of agents,
as studied by \citet*{Grundy1989} and many others. In each period,
the market will shut down with probability $1-\beta$, and the current
old agents will consume a terminal value of the asset. If it does
not shut down, the old will sell the asset to the young. In each period,
the young will inherit the distribution of signals about the terminal
value of the old. The asset price will equal the forward looking risk-adjusted
iterated expectation of the value of the asset. If the variance of
noise traders in the market increased without bound, there would be
no learning in the market and the expected risk-adjusted price would
be equal to the iterated average expectation. \citet*{allen2006beauty}
show this for a finite truncation of this environment with $\beta=1$.
The dynamic CARA-normal rational expectations model is studied under
the common prior assumption. \citet*{banerjee-kremer} and \citet*{han16}
have studied the role of heterogeneous prior beliefs in the static
version of the model. 

This asset market combines features that appear in many other asset
pricing models, and we now review some of the connections. \citet*{HarrisonKreps}
study an asset market where an asset is re-traded in each period between
different risk-neutral agents with heterogeneous prior beliefs. They
focus on the minimal price paths, in order to rule out bubbles based
purely on everyone's expectation that prices will rise based on calendar
time; we achieve a similar effect with our stationarity assumption.
We allow asymmetric information but make exogenous the agent to whom
another agent must sell. \citet*{duffie-manso} study a random matching
model of trade, where traders are matched in pairs at each time period.
They focus on information percolation over time with a simple updating
rule, while we focus on effects due to higher-order beliefs; our matching
technology is also more general. \citet*{malamud-rostek} study markets
with an exogenous network structure of access to multiple markets,
but endogenize agents' choice of how much to trade in each market. 

A key simplification in our model of trading is that each agent is
infinitesimal, so any learning about the asset value does not affect
anyone's expectations. \citet*{steiner2014price} obtain the same
effect in a model of asymmetric information where agents do not condition
on others' information. They give a behavioral interpretation of this
restriction via coarse perceptions.\footnote{One can also give their results an interpretation in terms of heterogeneous
beliefs and asymmetric information.} Our model and that of \citet*{steiner2014price} both feature the
same dependence of prices only on public information among the agents;
the limit where trading becomes frequent is critical to this. 

\section{The Interaction Structure \label{sec:The-Interaction-Structure}}

\subsection{Interaction Structure\label{sec:interaction-structure-intro}}

One contribution of this paper is to show that the information structure
and the network structure can be seen from a unified perspective\textemdash in
studying higher-order average expectations and, consequently, for
our applications. In particular, we will define an \emph{interaction
structure}\textemdash a square matrix indexed by the set $S$ comprising
the union of everyone's signals\textemdash that simultaneously captures
beliefs and the network. This serves two purposes. First, it highlights
the symmetry between information and the network. Second, it facilitates
relating higher-order average expectations to a Markov matrix and
its iteration, which is an important technique for us.\footnote{\citet*{samet1998iteratedA} introduced and used a Markov process
as a representation of an information structure. We construct a related,
but different, process: Ours simultaneously captures the network and
agents' beliefs and operates on the union of signals instead of realizations.
See \citet*{GolubMorris2016} for the exact analogue of Samet's process.} Indeed, we will use a Markov process representation to deduce the
results that follow from results about Markov processes. 

Let $S=\bigcup_{i\in N}T^{i}$ be the union of the (disjoint) sets
of signals.\footnote{Recall that this object appeared in the definition of joint connectedness
in Section \ref{sec:joint-connected}. It should not be confused with
the product set $T=\prod_{i\in N}T^{i}$, whose elements are signal
\emph{profiles}.} Define $x(n):S\to\mathbb{R}$ by $[x(n)](t^{i})=[x^{i}(n)](t^{i})$.
In words, this one function is a parsimonious way of keeping track
of the higher-order average expectations of \emph{all} agents at stage
$n$. A random variable $y:\Theta\to\mathbb{R}$ that depends on the
external state is viewed as a vector indexed by $\Theta$, i.e., $y\in\mathbb{R}^{\Theta}$.
The first-order expectation map $y\mapsto x(1)$ can then be viewed
as a map $\mathbb{R}^{\Theta}\to\mathbb{R}^{S}$. Using the standard
bases for the domain and codomain, we can represent this map via a
matrix. Indeed, we can write $x(1)=Fy$, where $F$ is a matrix with
rows indexed by $T^{i}$ and columns indexed by $\Theta$, and whose
entries are 
\begin{equation}
F(t^{i},\theta)=\pi^{i}(\theta\mid t^{i}).\label{eq:F-matrix}
\end{equation}
Even though the rows and columns of this matrix are not ordered, we
can define matrix multiplication by stipulating that 
\[
(Fy)(t^{i})=\sum_{\theta\in\Theta}F(t^{i},\theta)y(\theta).
\]
It is immediate to check that with this definition, $(Fy)(t^{i})$
is indeed $i$'s subjective expectation of $y$ when $i$ receives
signal $t^{i}$.

Along the same lines, the formula of (\ref{eqn:def-iteration}), $x^{i}(n+1;y)=\sum_{j\in N}\gamma^{ij}E^{i}x^{j}(n;y)$,
can be described in matrix notation. Equation (\ref{eqn:def-iteration})
defines a linear map $\mathbb{R}^{S}\to\mathbb{R}^{S}$ such that
$x(n)\mapsto x(n+1)$. Taking the standard basis for $\mathbb{R}^{S}$
(as both the domain and codomain) we can write $x(n+1)=Bx(n)$, where
$B$ is a matrix with rows and columns indexed by $S$, and entries
\begin{equation}
B(t^{i},t^{j})=\gamma^{ij}\pi^{i}(t^{j}\mid t^{i}).\label{eqn:b-def}
\end{equation}
We call $B$ the \emph{interaction structure}. It captures the weights
(arising from both the network and agents' beliefs) that matter for
iterating agents' expectations. 

Combining the above, we find, for $n\geq1$, the short formula 
\begin{equation}
x(n)={B}^{n-1}{F}y,\label{eqn:short-xn}
\end{equation}
which describes the step-$n$ higher-order average expectations. Thus,
understanding their behavior boils down to studying powers of the
linear operator $B$. One can check that:
\begin{fact}
The interaction structure $B$ is row-stochastic. 
\end{fact}
To verify this, note that for each $t^{i}\in S$ we have 
\[
\sum_{t^{j}\in S}B(t^{i},t^{j})=\sum_{j\in N}\sum_{t^{j}\in T^{j}}\gamma^{ij}\pi^{i}(t^{j}\mid t^{i})=\sum_{j\in N}\gamma^{ij}\sum_{t^{j}\in T^{j}}\pi^{i}(t^{j}\mid t^{i})=1.
\]
The final equality for each $t^{i}$ follows because the distribution
$\pi^{i}(\cdot\mid t^{i})$ is a probability distribution over $T^{j}$
and $\Gamma$ is row-stochastic. 

We will occasionally emphasize the dependence of the matrices we have
defined on $\bm{\pi}=(\pi^{i}(\cdot\mid t^{i}))_{t^{i}\in T^{i},i\in N}$,
and the dependence of $B$ on the network $\Gamma$, by writing $F_{\bm{\pi}}$
and $B_{\bm{\pi},\Gamma}$, and similarly for derived objects.

The interaction structure $B$ allows us to recover a matrix corresponding
to one agent's beliefs about another. For any $i$ and $j$, if we
set $\gamma^{ij}$ to $1$ and all the other entries of $\Gamma$
to $0$, then $B$ restricts naturally to an operator $B^{ij}:\mathbb{R}^{T^{j}}\to\mathbb{R}^{T^{i}}$
sending $T^{j}$-measurable random variables to $i$'s conditional
beliefs about them. The entries of the matrix are $B^{ij}(t^{i},t^{j})=\pi^{i}(\cdot\mid t^{i})$.

Equation (\ref{eqn:short-xn}) entails a sharp separation between
(i) agents' first-order beliefs about $\Theta$, on the one hand,
and (ii) the network and their beliefs about each other's signals,
on the other. The former are encoded in $F$, and the latter in $B$. 

\subsection{The Consensus Expectation via the Interaction Structure\label{sec:consensus-expectation}}

In Section \ref{subsec:consensus-expectation}, we defined the consensus
expectation. The formalism we have introduced will allow us to prove
Proposition \ref{prop:char}, below, on its existence, and in the
process also to relate it to properties of the matrix $B$.

Recalling Definition \ref{def:consensus-expectation}, the consensus
expectation is the number in every entry of the following vector:
\begin{equation}
\lim_{\beta\uparrow1}\left(1-\beta\right)\sum_{n=1}^{\infty}\beta^{n-1}x(n;y)\label{limit}
\end{equation}
The notation introduced in Section \ref{sec:interaction-structure-intro}
above allows us to rewrite this as
\begin{equation}
\lim_{\beta\uparrow1}\left(1-\beta\right)\left({\displaystyle \sum\limits _{n=0}^{\infty}\beta^{n}B^{n}}\right)Fy\text{.}\label{eq:powers-Bn}
\end{equation}
In this section, we use the formalism we have introduced to explain
why this limit exists and why it is a constant vector, as well as
to characterize it. The following is our main result on this, which
shows that the consensus expectation (recall Definition \ref{def:consensus-expectation}
in Section \ref{subsec:consensus-expectation}) is well-defined. 
\begin{prop}
\label{prop:char}The consensus expectation exists and
\begin{equation}
c(y;\bm{\pi},\Gamma)=\sum_{t^{i}\in S}p(t^{i})E^{i}[y\mid t^{i}],\label{eq:c-in-terms-of-p}
\end{equation}
where $p$ is the unique vector in $p\in\Delta(S)$ satisfying $pB=p$.
All entries of $p$ are positive, and it is called the vector of \emph{agent-type
weights}.
\end{prop}
Thus the consensus expectation of $y$ is a weighted average of the
expectations associated with the various signals of each agent, encoded
in $Fy$; the weight on the expectation of signal $t^{i}$ of agent
$i$, or simply type $t^{i}$, is given by $p(t^{i})$.\footnote{Recalling Section \ref{sec:info-structure}, we use the terminology
of a \emph{type} here for a signal to emphasize the interim perspective:
All that matters for higher-order expectations (and hence consensus
expectations) are an agent's interim beliefs (including higher-order
beliefs), and agents' types fully capture these.} Note that, by definition, $p$ is the stationary distribution of
$B$ viewed as a Markov matrix. 

A simple but important separation can be read off from the formula
of Proposition \ref{prop:char}. The vector $p$, because it is uniquely
defined by $B$ (by the Perron\textendash Frobenius Theorem), depends
only on the entries of $B$, which in turn depend only on the network
weights $\gamma^{ij}$ and on agents' interim marginals on one another's
signals, $\pi^{i}(t^{j}\mid t^{i})$. Thus, these features of the
model jointly determine the weights $p(t^{i})$. Beliefs about $\Theta$
enter only through $E^{i}[y\mid t^{i}]$. This reflects the separation
noted at the end of Section \ref{sec:interaction-structure-intro}.
Thus, the interesting effects arising from higher-order beliefs will
be characterized by explaining how the information structure affects
$p$; see, for instance, Sections \ref{sec:optimism} and \ref{sec:Ignorance}.

For our analysis, we have fixed a $y$ throughout; however, note that
if $y$ were arbitrary, Proposition \ref{prop:char} would hold with
the \emph{same} $p$ for all $y$. \medskip{}

To see why Proposition \ref{prop:char} holds, first note that if
\begin{equation}
\lim_{n\to\infty}B^{n}Fy\label{simple limit}
\end{equation}
exists, then this limit will equal (\ref{limit}). This is because
(\ref{limit}) is the weighted mean of terms of the form $B^{n}Fy$;
as $\beta\uparrow1$, most of the weight is assigned to the terms
corresponding to large values of $n$. To give intuition, here will
assume that (\ref{simple limit}) exists, though our result is more
general as shown in the proof of Proposition \ref{prop:char} in Appendix
\ref{subsec:Existence-and-Characterization-Proof}.\footnote{Sometimes (\ref{limit}) will exist when (\ref{simple limit}) does
not, because for large $n$, the vector $B^{n}Fy$ cycles (approximately)
among several limit vectors. In this case, (\ref{limit}) takes an
average of these vectors. We discuss these issues further in Appendix
\ref{periodicity-1}.} 

Joint connectedness will imply that the matrix $B$ is irreducible.\footnote{The meaning of irreducibility in our context is discussed further
in Section \ref{sec:irred}.} Thus, by a standard fact about such matrices, every row of $B^{\infty}$
is $p$, assuming this limit exists. Writing $\mathbf{1}$ for the
function (vector in $\mathbb{R}^{S}$) that takes a constant value
of $1$ on all of $S$, for any vector $z\in\mathbb{R}^{S}$ we have
\begin{equation}
\lim_{\beta\uparrow1}\left(1-\beta\right){\displaystyle \sum\limits _{n=0}^{\infty}\beta^{n}B^{n}z}=(pz)\mathbf{1},\label{eq:abel}
\end{equation}
where $p$ is as defined in the statement of Proposition \ref{prop:char}.
In the analysis of (\ref{limit}), we set $z=Fy$. A variant of the
standard Markov chain result shows that (\ref{eq:abel}) holds more
generally, even when the limit of the $x(n)$ in (\ref{eqn:short-xn})
does not exist. 

Proposition \ref{prop:char} implies that higher-order average expectations
converge in the sense of (\ref{limit}) to a number which is independent
of the agent \emph{and} of his signal: the consensus expectation.
Thus, in the coordination game, agents' actions in the $\beta\uparrow1$
limit, where coordination concerns dominate, are equal to a nonrandom
consensus.\footnote{If there are public events, the consensus is nonrandom once public
information is taken into account. See Section \ref{subsec:Consensus-without-Irreducibility}
for further discussion. } Of course, the consensus expectation depends, in general, on all
the interim beliefs $(\pi^{i})_{i\in N}$ and on the network $\Gamma$. 

\subsection{A Markov Process Interpretation of the Interaction Structure and
the Consensus Expectation\label{sec:physical}}

The interaction structure $B$ is a row-stochastic or Markov matrix,
and corresponds to a Markov process that we construct, with $S$ playing
the role of the state space. We can imagine a particle starting at
some state $t^{i}\in S$, and the probability of transitioning to
$t^{j}\in S$ being $\gamma^{ij}\pi^{i}(t^{j}\mid t^{i})$. 

This process can be useful for understanding the behavior of higher-order
average expectations. Fix a signal $t^{i}\in S$ and consider the
Markov process started at this $t^{i}$, with its (random) location
over time captured by the random variables $W_{1}=t^{i},W_{2},W_{3},\ldots$.
If we define a function $f:S\to\mathbb{R}$ such that $f(t^{i})=(Fy)(t^{i})$,
then \emph{$x^{i}(n)$, the $n^{\text{th}}$-order average expectation
of $y$, is the expected value of $f(W_{n})$}. The vector of agent-type
weights discussed in Section \ref{sec:consensus-expectation} is the
stationary distribution of the chain, and the consensus expectation
of $y$ is the expected value of $f(W)$ where $W$ is drawn according
to the stationary distribution.

The process we have defined provides a physical analogy that is useful
for intuition and also suggests proof techniques\textemdash see Sections
\ref{sec:optimism} and \ref{sec:Ignorance}. 

\section{The Consensus Expectation and the Network\label{sec:consensus-and-network}}

One simple special case of Proposition \ref{prop:char} arises when
$|T^{i}|=1$ for each $i$: There is complete information about each
agent's signal. In that case, $B=\Gamma$ and so $p=e$, the eigenvector
centrality vector of the network $\Gamma$. (Recall Definition \ref{def:eigenvector-centrality}
in Section \ref{sec:network}.) It follows from (\ref{eq:c-in-terms-of-p})
that
\[
c(y;\bm{\pi,}\Gamma)=\sum_{i}e^{i}E^{i}y,
\]
where, abusing notation, $E^{i}y$ denotes the interim expectation
of $y$ induced by the one signal that agent $i$ ever gets. This
relates to network game results of \citet*{Ballester2006}, and especially
to the limit with high coordination motives studied in \citet*{cadp15},
where play is determined by ideal points weighted by eigenvector centralities.

There is a much more general sense in which the eigenvector centralities
of the agents figure in the consensus expectation:
\begin{prop}
\label{prop:representation}There are strictly positive priors $\left(\lambda^{i}\right)_{i\in N}$,
with \textup{$\lambda^{i}\in\Delta\left(T^{i}\right)$, }\textup{\emph{such
that}}, for all $y$,
\begin{equation}
c(y;\bm{\pi,}\Gamma)=\sum_{i}e^{i}\mathbf{E}^{\lambda^{i}}y,\label{eq:representation}
\end{equation}
where the $e^{i}$ are the eigenvector centralities of the agents.
\end{prop}
\smallskip{}

The expression $\mathbf{E}^{\lambda^{i}}y$ corresponds to an\emph{
}ex ante expectation of agent $i$, where the expectation is taken
according to a \emph{pseudoprior} $\lambda^{i}$ over $i$'s signals
that need not be related to agent $i$'s actual prior $\mu^{i}$ . 

Recalling (\ref{eq:c-in-terms-of-p}), we can see that this result
asserts $e^{i}\lambda^{i}(t^{i})=p(t^{i})$, and indeed its content
is that agent $i$'s agent-type weights sum to his eigenvector centrality,
$e^{i}$. This is formally stated in the following lemma, which is
what we use to prove Proposition \ref{prop:representation}, and which
also relates to the Markov interpretation of consensus expectations
in Section \ref{sec:physical}.
\begin{lem}
For each $i$, the agent-type weights associated with agent $i$'s
types add up to the eigenvector centrality of $i$:

\[
\sum_{t^{i}\epsilon T^{i}}p(t^{i})=e^{i}.
\]
\end{lem}
\begin{proof}
Let $\iota:S\to N$ map any type $t^{i}$ to the agent $i$ whose
type it is. Check that $V(n):=\iota(W(n))$ is a Markov process on
$N$ with transition matrix $\Gamma$.\footnote{The reasoning is as follows: The probability of the event $\{V(n+1)=j\}$
conditional on $\{V(n)=i\}$ is equal to $\gamma^{ij}$: For any $t^{i}\in T^{i}\subseteq S$,
we have 
\[
\sum_{t^{j}\in T^{j}}B(t^{i},t^{j})=\sum_{t^{j}\in T^{j}}\gamma^{ij}\pi^{i}(t^{j}\mid t^{i})=\gamma^{ij}\sum_{t^{j}\in T^{j}}\pi^{i}(t^{j}\mid t^{i})=\gamma^{ij}.
\]
} Now the stationary probabilities of the process $W$ are given by
$p$, and the total stationary probability of the set $T^{i}\subseteq S$
under $W$ is therefore $\sum_{t^{i}\in S}p(t^{i})$. By the coupling
between $V$ and $W$, this must be equal to the stationary probability
of $i$ under $V$, which is $e^{i}$. 
\end{proof}
The proof of Proposition \ref{prop:representation} is completed by
making the definition $\lambda^{i}(t^{i})=p(t^{i})/e^{i}$, which
is legitimate because all the centralities $e^{i}$ are positive (see
comments after Definition \ref{def:eigenvector-centrality}).

Generally, the pseduopriors $\lambda^{i}$ will depend on \emph{both}
the information structure $\bm{\pi}$ and the network $\Gamma$. We
will be especially interested in when the $\lambda^{i}$ depend only
on beliefs. The next section gives some conditions for this, and the
issue is discussed more generally in Section \ref{subsec: samet}.

\subsection{Interpreting the Interaction Structure as a Network\label{subsec:Interpreting-the-Interaction-Structure}}

As we mentioned at the end of Section \ref{subsec:game}, in the context
of the interaction game, the weights $\gamma^{ij}$ can be interpreted
as $i$'s subjective probabilities of meeting or interacting with
various others at the time he has to commit to his action. In light
of this interpretation, $B(t^{i},t^{j})=\gamma^{ij}\pi^{i}(t^{j}\mid t^{i})$
can be seen as a subjective probability assessed by agent $i$, when
he has signal $t^{i}$, that his partner in the game will have signal
$t^{j}$: The first factor, $\gamma^{ij}$, is $i$'s probability
of meeting $j$, and $\pi^{i}(t^{j}\mid t^{i})$ is the probability,
conditional on that meeting, that $j$ has signal $t^{j}$. (An agent
may be privately informed about his weights or interaction probabilities.
This kind of uncertainty relates to that studied by \citet*{galeotti2010network};
see our discussion in Section \ref{sec:type-dependent-gamma}.) 

In this sense, the environment can be reduced, from the perspective
of each player, purely to incomplete information. Relatedly, we can
reduce the analysis purely to networks. To this end, we construct
a new environment (whose objects are distinguished by hats) based
on the primitives of the original environment. In this environment
the new set of \emph{agents, $\widehat{N}$, }is $S$, the set of
all signals. The network is $\widehat{\Gamma}(t^{i},t^{j})=\gamma^{ij}\pi^{i}(t^{j}\mid t^{i})$;
there is complete information about signals (each agent has a singleton
type $t^{i}$, which is also his agent label); and the first-order
beliefs of the new agents replicate those of the corresponding types.
Now, the higher-order average expectation vector of this new environment,
$\widehat{x}(n;y)$, is the sam\textcyr{\cyre}\footnote{Under the obvious bijection of indices.}
as $x(n,y)$. All statements about higher-order average expectations
in the original game of incomplete information can be reinterpreted
in this complete-information environment as network quantities. For
instance, to get the second-order average expectation of a type $t^{i}$,
we look at the corresponding agent in the network, and take the average,
across all his neighbors, of their neighbors' first-order expectations. 

To summarize: We have taken all uncertainty about others' signals,
and combined it with the original network weights, to obtain the new
network weights $\widehat{\Gamma}$. From this perspective, the game
of incomplete information of Section \ref{subsec:game} is reduced
to the network game studied by \citet*{Ballester2006}. This transformation
is essentially the transformation of the game of incomplete information
into an agent normal form. (For a conceptually similar reduction,
see \citet*{morris1997interaction}. The tensor products of \citet*{Marti2015}
can also be seen as instances of this in a specific setting of exchangeable
information.) 

\section{Unifying and Generalizing Network and Asymmetric Information Results\label{sec:cpa-signals}}

We now study conditions under which the agent-type weights take a
particularly simple form. Under these conditions, there are formulas
for consensus expectations that decompose nicely into different individuals'
prior expectations, weighted by those individuals' centralities.

Recall that agents' priors are given by the profile $(\mu^{i})_{i\in N}$
of distributions, with $\mu^{i}\in\Delta(T^{i})$.
\begin{defn}
\label{def:cpa-signals}There is a \emph{common prior over signals}
(CPS) if, for each signal profile $t\in T$ and each $i,j\in N$,
we have 
\[
\mu^{i}(t^{i})\pi^{i}(t^{-i}\mid t^{i})=\mu^{j}(t^{j})\pi^{j}(t^{-j}\mid t^{j}).
\]
\end{defn}
CPS does not imply a common prior over the states $\Theta$; agents
may have inconsistent beliefs about $\theta$. A common prior on signals
could arise if each agent first observed a signal drawn according
to the common prior but interpreted signals differently. However,
CPS does imply that there is a common prior over agents' second-order
and higher-order beliefs.

Now we can show that under CPS, the distributions $\lambda^{i}$ in
the representation $c(y;\bm{\pi,}\bm{\Gamma})=\sum_{i}e^{i}\mathbf{E}^{\lambda^{i}}y$
of Proposition \ref{prop:representation} are ex ante probability
distributions on signals, i.e., $\lambda^{i}=\mu^{i}$; the pseudopriors
are the \emph{actual} priors. Recall from Section \ref{subsec:Priors}
that bold expectation operators denote ex ante expectations. 
\begin{prop}
If there is a common prior over signals, then the consensus expectation
is equal to the eigenvector-centrality weighted average of the ex
ante expectations of the agents:
\[
c(y;\bm{\pi,}\Gamma)=\sum_{i}e^{i}\mathbf{E}^{\mu^{i}}y,
\]
where $\mu^{i}$ is the prior over $i$'s signals.\label{prop:cpa-signals}
\end{prop}
Proposition \ref{prop:cpa-signals} shows that the consensus expectation
is a weighted average of agents' prior expectations, $\mathbf{E}^{\mu^{i}}y$,
weighted by agents' network centralities, $e^{i}$. We say in this
case that there is a \emph{separability} between the network and the
information structure: The network enters only into the centralities,
and the information structure determines $\mathbf{E}^{\mu^{i}}y$.
(See Section \ref{subsec: samet} for further discussion of this property.)
Under complete information about signals but heterogeneous priors
about $\Theta$, this yields a reinterpretation of the DeGroot model,
as we discuss further in Section \ref{de groot-1}.

In terms of the generality of the information structure, Proposition
\ref{prop:cpa-signals} goes beyond previous related results that
decomposed equilibrium actions into agent-specific quantities weighted
by agents' centralities. Results in this category include \citet*{cadp15},
\citet*{behm15} and \citet{myatt-wallace} (which rely on Gaussian
signals), \citet*{Marti2015} (which relies on exchangeable signals),
and \citet*{bbdj15} (which does not characterize the contribution
of higher-order expectation terms). Formal details of each of these
models differ in several ways from our model, but in not imposing
parametric or symmetry conditions, and in allowing heterogeneous priors
about states, our result on the decomposition at the $\beta\uparrow1$
limit is more general than the others. 

If each ex ante expectation is the same, in Proposition \ref{prop:cpa-signals}
it does not matter what the eigenvector centralities are, and we have:
\begin{cor}
\label{cor:full-cpa}If there is a common prior over signals and $\mathbf{E}^{\mu^{i}}y=\overline{y}$
for all $i$\textemdash that is, agents have a common ex ante expectation
of the external random variable\textemdash then the consensus expectation
is equal to the (common) ex ante expectation.
\[
c(y;\bm{\pi,}\Gamma)=\overline{y}.
\]
\end{cor}
Corollary \ref{cor:full-cpa} is closely related to \citet*{samet1998iteratedA},
which shows that if the common prior assumption (over the whole space
$\Omega$) holds, then any sequence of expectations ($A$'s expectation
of $B$'s expectation $\ldots$) of the random variable is equal to
the ex ante expectation of the random variable $\overline{y}$. Since
limits of such iterated expectations determine the consensus expectation,
it is also equal to $\overline{y}$. Note, however, that the hypotheses
of Corollary \ref{cor:full-cpa} are weaker than the full common prior
assumption, because they impose no restrictions on the joint distribution
of $\theta$ and signals.

Proposition \ref{prop:cpa-signals} is also closely related to \citet*{samet1998iteratedA}
in the following sense: given a common prior over signals, the highly
iterated expectation of any random variable measurable with respect
to agent $i$'s signal is equal to the common prior expectation of
that random variable\textemdash that is, the expectation of it with
respect to the measure $\mu^{i}$. Our results show that these prior
expectations are combined according to agents' network centrality
weights. Sections \ref{subsec: samet} and \ref{subsec:Interim-Interpretation}
elaborate further on these issues, as well as a converse to Samet's
result.

\section{Contagion of Optimism\label{sec:optimism}}

Consider a case in which agents are \emph{second-order optimistic}:
they are optimistic about the expectations of those they interact
with. That is, they believe that, on average, those others have higher
expectations than their own. In this circumstance, we will give conditions
under which consensus expectations are driven to extremes via a contagion
of optimism. Sections \ref{subsec:game-applications-of-main-results}
and \ref{subsec:Motivation} state the interpretation of this in the
game and in the asset market, respectively.

\subsection{Three Illustrative Cases\label{subsec:Three-Illustrative-Cases}}

To motivate our results on this and to gain intuition, we first consider
some extreme cases. These illustrate how the Markov process representation
of higher-order expectations and its physical interpretation from
Section \ref{sec:physical} can yield striking results about consensus
expectations. Fix a random variable with minimum realization $0$
and maximum realization $1$. Say that agent \emph{$i$ considers
$j$ over-optimistic} if agent $i$'s expectation of agent $j$'s
expectation is always strictly greater than his own expectation (unless
his own expectation is $1$, in which case he is sure that agent $j$'s
expectation is $1$). Say that agent $i$ thinks that agent $j$ is
over-pessimistic if agent $i$'s expectation of agent $j$'s expectation
is always strictly \emph{less} than his own expectation (unless his
own expectation is $0$, in which case he is sure that agent $j$'s
expectation is $0$). 
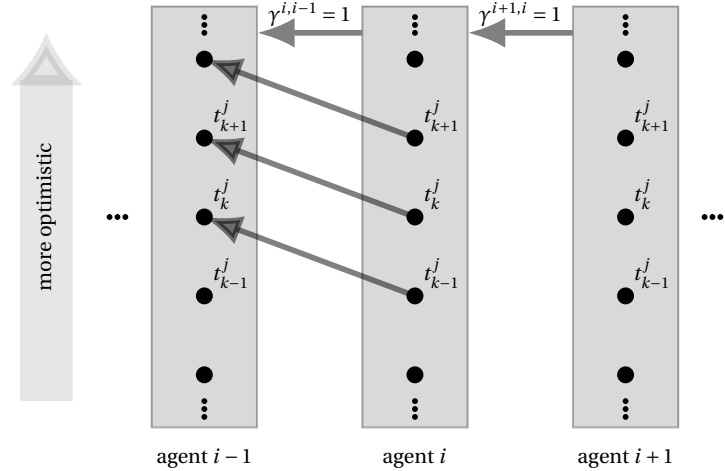
\begin{figure}
\begin{centering}
\begin{tikzpicture}[thick,scale=0.7, every node/.style={transform shape}]
\begin{scope}[transparency group, opacity=0.3]\draw [fill=gray, thick] (0,-2.5) rectangle (2,5.5); \end{scope} \begin{scope}[transparency group, opacity=0.3]\draw [fill=gray, thick] (4,-2.5) rectangle (6,5.5); \end{scope} \begin{scope}[transparency group, opacity=0.3]\draw [fill=gray, thick] (8,-2.5) rectangle (10,5.5); \end{scope}
\filldraw[black] (1,-2) circle (1pt) ; \filldraw[black] (1,-2.15) circle (1pt) ; \filldraw[black] (1,-2.3) circle (1pt) ;
\coordinate [label=below:{agent $i-1$}] (a) at (1,-2.75);  
\filldraw[black] (5,-2) circle (1pt) ; \filldraw[black] (5,-2.15) circle (1pt) ; \filldraw[black] (5,-2.3) circle (1pt) ; \coordinate [label=below:{agent $i$}] (a) at (5,-2.75);  
\filldraw[black] (9,-2) circle (1pt) ; \filldraw[black] (9,-2.15) circle (1pt) ; \filldraw[black] (9,-2.3) circle (1pt) ; \coordinate [label=below:{agent $i+1$}] (a) at (9,-2.75);  
\filldraw[black] (1,5) circle (1pt) ; \filldraw[black] (1,5.15) circle (1pt) ; \filldraw[black] (1,5.3) circle (1pt) ;
\filldraw[black] (5,5) circle (1pt) ; \filldraw[black] (5,5.15) circle (1pt) ; \filldraw[black] (5,5.3) circle (1pt) ;
\filldraw[black] (9,5) circle (1pt) ; \filldraw[black] (9,5.15) circle (1pt) ; \filldraw[black] (9,5.3) circle (1pt) ;
\filldraw[black] (-.5,1.5) circle (1pt) ; \filldraw[black] (-.65,1.5) circle (1pt) ; \filldraw[black] (-.8,1.5) circle (1pt) ;
\filldraw[black] (10.5,1.5) circle (1pt) ; \filldraw[black] (10.65,1.5) circle (1pt) ; \filldraw[black] (10.8,1.5) circle (1pt) ;
\filldraw[black] (1,-1.5) circle (4pt) ; \filldraw[black] (1,0) circle (4pt) node[anchor=south west] {$t_{k-1}^j$}; \filldraw[black] (1,1.5) circle (4pt) node[anchor=south west] {$t_k^j$}; \filldraw[black] (1,3) circle (4pt) node[anchor=south west] {$t_{k+1}^j$}; \filldraw[black] (1,4.5) circle (4pt) ;
\filldraw[black] (5,-1.5) circle (4pt) ; \filldraw[black] (5,0) circle (4pt) node[anchor=south west] {$t_{k-1}^j$}; \filldraw[black] (5,1.5) circle (4pt) node[anchor=south west] {$t_k^j$}; \filldraw[black] (5,3) circle (4pt) node[anchor=south west] {$t_{k+1}^j$}; \filldraw[black] (5,4.5) circle (4pt) ;
\filldraw[black] (9,-1.5) circle (4pt) ; \filldraw[black] (9,0) circle (4pt) node[anchor=south west] {$t_{k-1}^j$}; \filldraw[black] (9,1.5) circle (4pt) node[anchor=south west] {$t_k^j$}; \filldraw[black] (9,3) circle (4pt) node[anchor=south west] {$t_{k+1}^j$}; \filldraw[black] (9,4.5) circle (4pt) ;
\begin{scope}[transparency group, opacity=0.5] \draw[line width=.75mm,-{Latex[length=5mm, width=4mm]}] (5, 0) -- (1, 1.5); \end{scope} \begin{scope}[transparency group, opacity=0.5] \draw[line width=.75mm,-{Latex[length=5mm, width=4mm]}] (5, 1.5) -- (1, 3); \end{scope} \begin{scope}[transparency group, opacity=0.5] \draw[line width=.75mm,-{Latex[length=5mm, width=4mm]}] (5, 3) -- (1, 4.5); \end{scope}
\begin{scope}[transparency group, opacity=0.1] \draw[line width=7mm,-{Latex[length=7mm, width=10mm]}] (-2, -2) -- (-2, 5); \end{scope} \coordinate[label={[rotate=90]center:{more optimistic}}] (v3) at  (-2,1.5);
\draw[gray, line width=.75mm,-{Latex[length=5mm, width=4mm]}] (8, 5) -- (6, 5);  \coordinate[label={above:{$\gamma^{i+1,i}=1$}}] (v3) at  (7,5);
\draw[gray, line width=.75mm,-{Latex[length=5mm, width=4mm]}] (4, 5) -- (2, 5);  \coordinate[label={above:{$\gamma^{i,i-1}=1$}}] (v3) at  (3,5);   \end{tikzpicture}
\par\end{centering}
\caption{The example of Case II, with a counterclockwise network.}
\label{fig:case-ii}
\end{figure}

\subsubsection*{Case I}

First, suppose that each agent considers every other agent over-optimistic.\footnote{This\textemdash and other examples we describe here\textemdash may
involve violating the otherwise maintained joint connectedness assumption
(irreducibility of $B$), but our main result in this section, Proposition
\ref{prop:over-optimism}, does not rely on the joint connectedness
of $B$. } In this case the consensus expectation must be $1$, independent
of the network structure.

\subsubsection*{Case II}

Second, suppose that for every agent $i\text{,}$ there is an agent
he considers over-optimistic and another agent he considers over-pessimistic.
Then there is a network structure under which the consensus expectation
is $1$. We can simply look at the network structure in which each
agent puts all weight on agents he thinks are over-optimistic. Symmetrically,
there is a network structure in which the consensus expectation is
$0\text{.}$ These results do not depend on agents' ex ante expectations\textemdash which
might take any value between $0\text{ and \ensuremath{1\text{. }}}$

Figure \ref{fig:case-ii} illustrates one example of this occurring.
There are $I$ agents, indexed by $N=\left\{ 1,\ldots,I\right\} $,
and their indices are interpreted modulo $I$. Each agent has many
signal realizations, $t_{k}^{i}$, with indices $k\in\left\{ 1,\ldots,K\right\} $,
with higher-$k$ signals inducing more optimistic first-order beliefs
about $y$. Assume that the most extreme signals lead to expectations
$1$ and $0$. Agent $i$, when he has signal $t_{k}^{i}$, is certain
that agent $i-1$ has signal $t_{k+1}^{i-1}$, the next more optimistic
signal. He is also certain that agent $i+1$ has signal $t_{k-1}^{i+1}$,
the next more pessimistic signal. If $k$ is already extreme (that
is, $k=1$ or $K$) then we replace $k+1$ (respectively, $k-1$)
by $k$ in the above description. 

Now the two networks considered are as follows. One has each agent
assigning all weight to the agent counterclockwise from him (i.e.,
to his left, as depicted in Figure \ref{fig:case-ii}). The other
network has each agent assigning all weight to the one clockwise from
him (i.e., to his right). Then in the counterclockwise network, the
consensus expectation is $1$, and in the clockwise network (not shown),
the consensus expectation is $0$.

\subsubsection*{Case III}

For our final case, rather than assuming any agent is over-optimistic
about any\emph{ }other, assume instead that each agent's expectation
of the \emph{average} expectation of others' expectations is greater
than his own expectation. As always, averages are taken with respect
to network weights, and ``greater'' is strict except in the case
where an agent's expectation is $1$. This constitutes a milder form
of over-optimism. Note that it is not implied by the assumptions we
imposed in either of the above results: While the condition of Case
III depends on the network (as in Case II), it allows for the possibility
that an agent is never over-optimistic about any other particular
agent (recalling that over-optimism is a condition uniform\emph{ }over
one's signals). Rather, which agent someone is over-optimistic about
may depend on his signal. But again, the consensus expectation is
$1$.

\subsubsection*{Markov Process Intuitions}

The results in each of the cases above can be established by using
the representation of higher-order average expectations via a Markov
process, which we presented in Section \ref{sec:physical}.

Let us begin by explaining Case I. If a particle makes transitions
over the states $S$ according to the Markov process, then at each
step it  moves toward strictly more optimistic types of agents, unless
it is already at a most optimistic type. Similar arguments can be
given for the other cases; see the proof in Section \ref{subsec:Markov-Chain-for}
for the general argument.\medskip{}

The cases discussed so far involve the unsatisfactory assumption that
some types are \emph{certain }that they are the most optimistic. It
will often be unreasonable for agents to hold such extreme beliefs,
or for the analyst to assume that they do. Thus, we wish to have a
result that is more quantitative and more robust. Also, the networks
involved in Case II are extreme, not allowing an agent to put even
small amounts of weight on others whom he does not consider over-optimistic
(or over-pessimistic). Our general results will relax all these assumptions.

The basic idea behind that generalization is clear: it follows from
the arguments above and continuity. But the details are subtle. Indeed,
what will be most interesting about the general results we obtain
is the nature of the conditions that are involved. How much second-order
pessimism can be permitted for the very optimistic agents without
losing the contagion of optimism? By relating the situation of second-order
optimism to a suitable Markov chain, we are able to give a precise
bound describing how strong second-order optimism (of ``most'' types)
must be relative to the pessimism about counterparties' beliefs permitted
for very optimistic agents.

\subsection{A General Case}

\label{subsec:optimism-general}We now weaken our assumptions on the
most optimistic types, and allow for the possibility that when agents
are maximally optimistic they assign only probability $1-\varepsilon$,
for some $\varepsilon>0$, to any given other being maximally optimistic.
But now we assume that when an agent is not maximally optimistic,
there is a uniform lower bound, $\delta$, on the degree of over-optimism.
With this weakening of our earlier assumptions, the above results
remain true with an error of order $\frac{\varepsilon}{\delta}$.
Among other things, this allows us to use a network with $\gamma^{ij}>0$
for all pairs $i,j$ in Case II above.

We state and prove a formal version of our claims in this general
case, and then discuss how the claims made about our illustrative
cases follow. Critically, in addition to demonstrating the continuity
in beliefs we needed, this result gives quantitative bounds on how
agents' interim over-optimism translates into the consensus outcome.
\begin{prop}
\label{prop:over-optimism}Consider an arbitrary information structure
$\bm{\pi}$ and an arbitrary network $\Gamma$ (i.e., drop for this
result the maintained assumption that $B$ is irreducible). Suppose
there exist \textup{$\overline{f}$} and $\delta>0$, $\varepsilon\geq0$
such that beliefs about neighbors are mildly optimistic in the following
sense:
\begin{enumerate}
\item Every type whose first-order expectation of $y$ is strictly below
$\overline{f}$ expects the first-order expectation, averaged across
his counterparties, to be at least $\delta$ above his own. That is,
for every $t^{i}$ such that $(E^{i}y)(t^{i})<\overline{f}$, we have
$\sum_{j}\gamma^{ij}(E^{i}E^{j}y)(t^{i})\geq(E^{i}y)(t^{i})+\delta$.
\item Every type whose first-order expectation of $y$ is at least \textup{$\overline{f}$}
expects the first-order expectation, averaged across his counterparties,
to be almost as large as his own, with a shortfall of at most $\varepsilon$.
That is, for every $t^{i}$ such that $(E^{i}y)(t^{i})\geq\overline{f}$,
we have $\sum_{j}\gamma^{ij}(E^{i}E^{j}y)(t^{i})\geq(E^{i}y)(t^{i})-\varepsilon$.
\end{enumerate}
Then the consensus expectation of $y$ is at least $\frac{\overline{f}}{1+\varepsilon/\delta}$.
\end{prop}
The proof of this result, via a suitable Markov chain inequality,
is provided in Section \ref{subsec:Markov-Chain-for} below.

An important feature of this result is that, fixing the constants
$\delta$ and $\varepsilon$, its hypotheses do not depend on the
finite type space used to represent the environment. This allows the
result to extend readily to infinite signal spaces, by considering
sequences of finite ones approximating the infinite one.

We now return to Cases I\textendash III, with expectations taking
values in $[0,1]$, and describe how to obtain them formally as applications
of this result. For Case I, where each agent considers every other
one over-optimistic, set $\overline{f}=1$. Because of finiteness
of the type space, there is a $\delta$ so that hypothesis (1) of
Proposition \ref{prop:over-optimism} holds for all types whose first-order
expectations of $y$ are strictly below $\overline{f}$. For this
case, we can take $\varepsilon=0$. Applying Proposition \ref{prop:over-optimism},
we get that the consensus expectation is $1$. For the case of over-pessimism,
we simply apply a change of variables from $y$ to $1-y$ and use
the same result to find that the consensus expectation is $0$.

For Case II, we constructed two networks. In one network, each agent
places all weight on some agent he considers over-optimistic. For
this network, the hypotheses of Proposition \ref{prop:over-optimism}
hold for the same reasons discussed in the previous paragraph, and
we conclude that the consensus expectation is $1$. In the other network,
each agent places all weight on some agent he considers over-pessimistic,
and by symmetry the consensus expectation is $0$. 

Case III is a direct application of the proposition, with $\varepsilon=0$.

\subsection{Markov Chain for Second-Order Optimism\label{subsec:Markov-Chain-for}}

We now analyze the interaction structures corresponding to second-order
optimism and discuss how to establish our results using Markov chain
arguments.

Consider an arbitrary finite state space $S$ with a Markov kernel,
with $B(s,s')$ being the probability of transitioning from state
$s$ to $s'$, and fix a function $f:S\to\mathbb{R}.$ 

The purpose of this subsection is to present the following lemma:
Assume that for all $s$ such that $f(s)$ is below a certain value
$\overline{f}$, taking one step from $s$ (according to the Markov
kernel) to reach a random state $W_{2}$ yields a value $f(W_{2})$
that is higher by at least $\delta$, in expectation, than $f(s)$.
Assume also that if, in contrast, $s$ is chosen such that $f(s)$
exceeds $\overline{f}$, then the expected value of $f(W_{2})$ can
decrease relative to $f(s)$ by only a smaller amount, $\varepsilon$.
Under these assumptions, we will show that if $s$ is drawn from a
stationary distribution of $B$, the expectation of $f(s)$ is not
much below $\overline{f}$. The lemma we now state makes this quantitative
and precise.

We denote by $W_{1},W_{2},\ldots$ the stochastic process induced
by the Markov chain. The symbol $\mathbb{P}_{W_{1}\sim\nu}$ denotes
the probability measure corresponding to this process when $W_{1}$
is drawn according to a distribution $\nu$.\footnote{When $\nu$ is a point measure on $s$, we write $W_{1}=s$ in the
subscript as a shorthand.} The notation for expectations is analogous.
\begin{lem}
\label{lem:markov-optimism}Let $B$ be a Markov chain as described
above. Suppose there are real numbers $\delta,\varepsilon>0$ and
$\overline{f}$ such that the following hold:
\begin{enumerate}
\item For every $s$ such that $f(s)<\overline{f}$, we have $\mathbb{E}_{W_{1}=s}[f(W_{2})]\geq f(s)+\delta$.
\item For every $s$ such that $f(s)\geq\overline{f}$, we have $\mathbb{E}_{W_{1}=s}[f(W_{2})]\geq f(s)-\varepsilon$.
\end{enumerate}
Fix an arbitrary starting state, and let $p$ denote the ergodic distribution
over states that is reached starting from that state.\footnote{Note that the chain need not have a unique ergodic distribution, but
there is \emph{an} ergodic distribution reached from any initial state.} Then $p(s:f(s)\geq\overline{f})\geq\frac{1}{1+\varepsilon/\delta}$.
\end{lem}
The proof, which appears in Section \ref{sec:Omitted-Proofs}, uses
the fact that $f(W_{2})$ and $f(W_{1})$ have the same expectations
under the ergodic distribution, and uses the hypotheses of the lemma
in this equation to derive the desired inequality. With this result,
we can establish all the conclusions about the consensus expectation,
as discussed after Proposition \ref{prop:over-optimism} above.

\subsection{Discussion}

\subsubsection{Related Results\label{subsec:optimism-related}}

The result also relates to \citet*{HarrisonKreps}, who consider the
case where risk-neutral agents have heterogeneous beliefs (but symmetric
information) and trade and re-trade an asset through time. The asset
is always sold to the (endogenously) most optimistic agent at the
current history. The price is driven above the highest expectation
of the asset's value held by any agent. \citet*{HarrisonKreps} motivate
their exercise as a model of ``speculation,'' and our result has
a similar interpretation. In both cases, which agent is most optimistic
can vary: in their case, the identity of the most optimistic is determined
by the public history of the performance of the asset, whereas for
us it is because of asymmetric information.

A closely related paper is that of \citet*{izmalkov10}. They make
a primitive assumption similar to our assumption about optimism: Beliefs
are all distorted in the same direction. They consider two-agent,
two-action coordination games, and show that agents can be induced
to take any rationalizable action\textemdash including risk-dominated
ones\textemdash if the degree of optimism is high enough.\footnote{This observation illustrates a more general point of \citet*{weinstein2007structure}that
any rationalizable action can be made uniquely rationalizable if a
type is perturbed in the product topology. } 

Finally, \citet*{han16} report a ``contagious optimism'' result
in a CARA-normal asset pricing model. They study a \emph{static }CARA-normal
pricing game in which an agent, in equilibrium, conditions on the
information revealed by a counterparty's trading. While our game is
designed to pick up higher-order average expectations, their result
depends on different properties of higher-order expectations (certain
kinds of hierarchies in which agents wrongly assume common knowledge
of the mean of an asset value). However, they similarly show that
a small amount of optimism can give rise to arbitrarily high asset
prices. Another difference is that their result is written for the
two-agent case; any extension to many agents would require that the
network be uniform, because trade takes place in centralized markets.
Because they consider a world with normally distributed uncertainty,
there is no upper bound on first-order expectations, and this allows
contagious optimism to drive prices up without bound. 

\subsubsection{Tightness\label{subsec:optimism-tightness}}

We now construct a chain to show the bound of Lemma \ref{lem:markov-optimism}
is tight. This shows the sufficient condition for contagion of optimism
is tight: in at least some cases, it gives exactly the amount of second-order
optimism needed to guarantee high consensus expectations.

Consider a chain with states $t_{k}^{i}$ for $i\in\{1,2\}$ and $k\in\{0,\ldots,m\}$
and . Let $f(t_{k}^{i})=k$ and define, whenever $j\neq i$:
\[
B(t_{k}^{i},t_{\ell}^{j})=\begin{cases}
\delta & \text{if }\text{}\ell=k+1\leq m\\
1-\delta & \text{if }\text{}\ell=k<m\\
\varepsilon & \text{if }k=m,\text{}\ell=m-1\\
1-\varepsilon & \text{if }k=m,\text{}s'=m\\
0 & \text{otherwise}.
\end{cases}
\]
All the other entries are $0$. 

If we view $k$ as the ``height'' of the chain, it ascends a step
with probability $\delta$ when $k$ is in the interval $\left\{ 0,1,\ldots,m-1\right\} $,
and takes a step downward otherwise; if it is at height $k=m$, the
maximum, it moves with probability $\varepsilon$ to height $k=m-1$.
Otherwise, it stands still. While we have described $B$ as a Markov
process, it can be realized as an interaction structure. Our notation
suggests how to realize this chain as an interaction structure with
two agents, each having $m+1$ types.

It is easy to compute that this chain achieves the lower bound of
Lemma \ref{lem:markov-optimism}, and this example can easily be adjusted
to be irreducible. (See Section \ref{subsec:appendix-tightness} in
the Appendix for details.) Thus, it really is necessary that $\varepsilon$
(pessimism) be bounded as in the formula of the lemma relative to
the guaranteed ``optimistic drift'' $\delta$. 

\section{Tyranny of the Least-Informed\label{sec:Ignorance} }

In Proposition \ref{prop:cpa-signals}, we gave a sufficient condition
(common prior on signals) under which the consensus expectation is
the centrality-weighted average of agents' prior expectations. In
this section, we will find conditions on the information structure
under which the consensus expectation is (almost) equal to \emph{one}
agent's expectation. That is, rather than influence being shared according
to network centrality, it will all be allocated to one agent, in a
way that will depend on the information structure. In particular,
it will turn out to be the \emph{least informed }agent who accumulates
influence.

To motivate these results, we can again consider some extreme cases.
First, suppose that one agent is completely ignorant and has no private
information, while other agents know the state perfectly. The agents
other than the ignorant agent will have degenerate interim beliefs,
so nothing about their priors can matter for iterated expectations
or the consensus. Thus, if anyone's ex ante beliefs play a role in
determining consensus expectations, it must be those of the least
informed agent. It turns out that the consensus expectation is simply
\emph{equal} to the ignorant agent's prior expectation of $y$. A
simple way to see this is to note that, because the ex ante beliefs
of the informed agents don't matter, we may as well take them to be
equal to the prior of the ignorant agent; then the conclusion follows
by Proposition \ref{cor:full-cpa} on the common prior. By continuity,
our result continues to hold if the ignorant agent has \emph{almost}
no information and the other agents have \emph{almost} perfect information. 

Surprisingly, this conclusion remains true when the ignorant agent
is only \emph{relatively} ignorant, and when his beliefs are not public
as they were in the toy example. The ignorant agent may possess very
precise private information about the state. But if others have even
more precise (i.e., less noisy) private information, then their priors
will still not matter, and only the relatively ignorant agent's priors
will determine the consensus expectation.

We now present the statement and proof of the result, and then discuss
it and compare it with related results in Section \ref{subsec:tyranny-discussion}.

\subsection{Common Interpretation of Signals Framework}

Fix a complete $\Gamma$, i.e., one such that $\gamma^{ij}>0$ whenever
$i\neq j$. We specialize to a framework that we call \emph{common
interpretation of signals}, following the terminology of \citet*{Kandel}
and \citet*{Acemoglu2016}. There is a state $\theta\in\Theta$ that
is drawn by nature. Each agent receives conditionally independent
signals about it according to a full-support distribution $\eta^{i}(\cdot\mid\theta)\in\Delta(T^{i})$;
these distributions are common knowledge. However, the agents have
different full-support priors, $\rho^{i}\in\Delta(\Theta)$, over
the state space. Combined with the conditional distributions encoded
in the $\eta^{i}$, these uniquely define a prior distribution over
$\Theta\times T$. We denote by $\mathbf{E}^{\rho^{i}}$ the corresponding
prior expectation operator. These primitives also induce in each agent,
via Bayes' rule, an interim belief function; for each $t^{i}\in T^{i}$,
there is a distribution $\pi^{i}(\cdot\mid t^{i})$ over both the
state and over others' signals. 
\begin{defn}
\label{def: CIS_1}We say that $\eta^{i}$ is \emph{at most $\varepsilon$-noisy}
if: for every $\theta\in\Theta$, there is exactly one signal $t_{\theta}^{i}$
satisfying $\eta^{i}(t_{\theta}^{i}\mid\theta)\geq1-\varepsilon$,
and this $t_{\theta}^{i}$ also satisfies $\eta^{i}(t_{\theta}^{i}\mid\theta')\leq\varepsilon$
for all $\theta'\neq\theta$. 
\end{defn}
This condition requires that for any $\theta$, there is exactly one
signal $t^{i}$ that $i$ receives with very high probability conditional
on $\theta$ being realized; moreover, no two different $\theta,\theta'$
can be associated with the same such signal. 
\begin{defn}
\label{def:CIS_2}We say that $\eta^{i}$ is\emph{ uniformly at least
$\delta$-noisy} if, for every $\theta\in\Theta$ and $t^{i}\in T^{i}$,
the inequality $\eta^{i}(t^{i}\mid\theta)\geq\delta$ holds. 
\end{defn}
This condition says that each signal has at least $\delta$ probability
of being observed under each state, limiting the amount of information
that can be inferred from any signal.

\subsection{Sufficient Conditions for Tyranny of the Least-Informed}

Before stating the main proposition, we introduce some quantities
that will figure in it. Let $\gamma_{\min}=\min_{i\neq j}\gamma_{ij}$
be the smallest off-diagonal entry of $\Gamma$, which is positive
by assumption. Let $\rho_{\min}^{i}$ be the minimal probability assigned
to any $\theta\in\Theta$ by the prior $\rho^{i}\in\Delta(\Theta)$
of agent $i$. Let $\rho_{\min}=\min_{i}\rho_{\min}^{i}$ be the minimum
of all of these, across agents. Finally, let $y_{\max}=\max_{\theta\in\Theta}|y(\theta)|$.
\begin{prop}
\label{prop:ignorant-player}Suppose that for some $\delta\in(0,1)$
and $\varepsilon\in(0,1/2)$, 
\begin{enumerate}
\item $\eta^{1}$ is uniformly at least $\delta$-noisy
\item $\eta^{i}$ for all $i\neq1$ is at most $\varepsilon$-noisy.
\end{enumerate}
Then 
\begin{equation}
|c(y;B_{\bm{\pi}},F_{\bm{\pi}})-\mathbf{E}^{\rho^{1}}[y]|\leq\frac{4|\Theta||S|^{2}}{(\gamma_{\min}\rho_{\min})^{2}}\cdot\overline{y}\cdot\frac{\varepsilon}{\delta}.\label{eq:bound-ignorant-player}
\end{equation}
 
\end{prop}
This bound is designed for cases where $\varepsilon$ is much smaller
than $\delta$. It says that if agent 1's information is at least
$\delta$-noisy, while all others' information is quite precise (at
most $\varepsilon$-noisy), then the difference between the consensus
expectation of $y$ and agent $1$'s expectation of $y$ is small:
The upper bound is linear in $\varepsilon/\delta$. The constants
depend on the sizes of the state space and the signal space $S$,
and on the minimum network and belief weights in the denominator. 

We could formulate a version of Proposition \ref{prop:ignorant-player}
without requiring the rather strong assumption of full support of
the conditional distributions $\eta^{i}(\cdot\mid\theta)$ that is
implied by Proposition \ref{def:CIS_2}. This is discussed below in
Section \ref{subsec:Generalization-Beyond-Full-Support}, once we
have a bit more notation.

\subsection{Key Steps in the Proof of Proposition \ref{prop:ignorant-player}\label{subsec:Proof-of-Proposition-Ignorant}}

We will analyze the consensus expectation in the situation of Proposition
\ref{prop:ignorant-player} by analyzing the interaction structure
$B$ and its stationary distribution, $p$. Indeed, the analysis here
is intended as our main illustration of the value of reducing informational
questions to questions about the Markov chain corresponding to the
interaction structure.

The key insight in proving Proposition \ref{prop:ignorant-player}
is to construct an artificial signal structure $\boldsymbol{\widehat{\eta}}$
in which all agents except agent $1$ are certain of what $\theta$
is. This is done by rounding the signal probabilities $\eta^{i}(t^{i}\mid\theta)$
for $i\neq1$ to $0$ or $1$. Along with the priors $(\rho^{i})_{i\in N}$
over $\theta$ that are part of the setup, this induces an artificial
information structure $\widehat{\bm{\pi}}=(\widehat{\pi}^{i})_{i}$.
We let $\widehat{B}=B_{\widehat{\boldsymbol{\pi}},\Gamma}$. 

The proof then proceeds in three steps. First, we prove that $p$,
the stationary distribution of $B$, is well-approximated by that
of $\widehat{B}$, which is denoted by $\widehat{p}$. Second, we
claim that $\widehat{\bm{\pi}}$ can be viewed as having a common
prior (corresponding to agent 1's prior beliefs). This is because
only agent 1 is uncertain under $\widehat{\bm{\pi}}$ about $\theta$,
and so the ex ante beliefs of the others about $\theta$ can make
no difference; indeed, it can be shown that the other agents' interim
beliefs are compatible with agent $i$'s prior. Thus the consensus
expectation of $y$ under $\widehat{B}$ is equal to $\mathbf{E}^{\rho^{1}}[y]$.
Finally, we combine these facts to derive the proposition. We carry
out these steps below, deferring technical details to Appendix \ref{sec:tyranny-proofs}.

The key technique in this argument deserves some extra comment. In
the first step, where we approximate $p$ by $\widehat{p}$, we apply
a result of \citet*{Cho-Meyer} on perturbations of Markov chains.
This result, loosely speaking, says the following: As long as the
changes in weights in going from $\widehat{B}$ to $B$ are small
relative to the reciprocal of the \emph{maximum} \emph{mean first
passage time} (MMFPT\footnote{The MMFPT in the interaction structure $\widehat{B}$ is defined to
be the maximum expected time it takes to get from one state to another
in the physical process of Section \ref{sec:physical}. It is a measure
of the connectedness of $\widehat{B}$ as a network; in belief terms,
it is a measure of the maximum number of iterations required for there
to be contagion of higher-order beliefs between the two ``farthest''
states in $S$.}) of $\widehat{B}$, then $p$ is close to $\widehat{p}$. In our
application, the change in the interaction structure (corresponding
to interim beliefs about $\theta$ of the relatively informed agents
$i\neq1$ changing from ``slightly uncertain'' to ``fully certain'')
is of order $\varepsilon$, and that is why $\varepsilon$ appears
in the numerator of the bound in Proposition \ref{prop:ignorant-player}.
In the situation of Proposition \ref{prop:ignorant-player}, the MMFPT
is of order $1/\delta$, the inverse of the lower bound on the uninformed
agent's noise. (That is why $\delta$ appears in the denominator in
the bound of Proposition \ref{prop:ignorant-player}.) But the technique
we have outlined applies more broadly, in any setting where the size
of the perturbation to the interaction structure can be bounded relative
to the MMFPT. This could be used to weaken the assumptions of Proposition
\ref{prop:ignorant-player}, for example to cover cases where noise
does not have full support or the network is not complete\textemdash see
Section \ref{subsec:Generalization-Beyond-Full-Support} below. 

We carry out the details of the proof in Section \ref{sec:tyranny-proofs}.

\subsubsection{The Case Where No Player Is at Least $\delta$-Uncertain \label{subsec:Generalization-Beyond-Full-Support}}

Suppose we did not assume that player $1$ is at least $\delta$-uncertain,
which entails the strong assumption that there is a lower bound on
the conditional probability of seeing any one of his signals, given
any possible state. Then we would define the \emph{uncertainty}, $\delta$,
of player 1's information as the minimum nonzero value of $\eta^{i}(t^{i}\mid\theta)$
(as $t^{i}$ and $\theta$ range over all possibilities). Along the
same lines, we might wish to relax the assumption that $\Gamma$ is
complete, with every player putting weight on every other. We now
discuss how the general principles of our argument would go through
and the nature of the subtleties that would arise.

As mentioned in the sketch of the proof above, what really matters
in the proof is MMFPTs in $\widehat{B}$. Assuming $\widehat{B}$
is irreducible, we can still bound these in terms of $\delta$ even
with the weaker assumptions just discussed. But\textemdash as an examination
of our bounds on the MMFPT shows\textemdash the bounds will involve
path lengths in $\widehat{B}$: the number of steps in $\widehat{B}$
that must be taken to link any two states. Thus, rather than a bound
on the MMFPT in $\widehat{B}$ of order $\delta^{-1}$, which is what
we use in our result, we might have a bound of order $\delta^{-5}$.
The exponent will depend both on the information structure and on
$\Gamma$. In the end, this will translate into a difference on the
right-hand side of (\ref{eq:bound-ignorant-player}) in Proposition
\ref{prop:ignorant-player}. Indeed, we conjecture that the ratio
$\varepsilon/\delta$ would be replaced by $C\varepsilon/\delta^{\kappa}$
for a number $\kappa$ that is increasing in the maximum path length
in $\widehat{B}$. Moreover, this adjustment would be necessary: In
the more general setting we are discussing here, it is not possible
to write a bound analogous to (\ref{eq:bound-ignorant-player}) that
depends on $\varepsilon$ and our generalized $\delta$ only through
$\varepsilon/\delta$. 

While a full exploration of these elaborations is beyond the scope
of the present work, our point is to say: (i) the MMFPT technique
discussed here does cover less restrictive assumptions on information
than we made for our illustrative result; and (ii) the topology of
connections among types in the interaction structure $\widehat{B}$
will matter in interesting ways for more general results.

\subsection{Interpretation and Discussion\label{subsec:tyranny-discussion}}

\subsubsection*{Why Focus on the Least-Informed?}

The results of this section may seem paradoxical. In models of coordination
on a network motivated by organizational questions, a common result
is that agents have an incentive to focus on more informed agents,
in the sense of paying more attention to them or putting more weight
on their signals; see, for example, \citet*{cadp15}, \citet*{Herskovic15}.
Part of the reason for the difference in our result is that asymmetric
information gets washed out in our limit of higher-order expectations
(recall Proposition \ref{prop:char}), rather than being learned or
aggregated, and this makes the forces determining influence different.
In \citet{myatt-wallace}, the agents are choosing which \emph{signal
sources} to listen to (of a commonly available set) in a coordination
game; there publicness and clarity also play a role, though in different
ways.

\subsubsection*{The Least-Informed Become Effectively More Central.}

It is also interesting to compare the result on the tyranny of the
least-informed with the result of Proposition \ref{prop:cpa-signals}
in Section \ref{sec:cpa-signals}, where we showed that, under a consistency
condition on beliefs, it is an agent's centrality that determines
his influence. However, as is seen in our simple benchmark example
above, for sufficiently well-informed agents, their priors cannot
possibly matter, no matter how central they are in the network. Thus,
an implication of Proposition \ref{prop:cpa-signals} and Proposition
\ref{prop:ignorant-player} taken together is that the conditions
of Proposition \ref{prop:ignorant-player} cannot, in general, be
reconciled with common priors over beliefs/signals.

We can get some further intuition for our result by expressing it
in the language of our applications. Suppose that agents are making
investment decisions, but with strategic complementarities in those
decisions. We might say that there is \emph{confidence} in the economy
if positive expectations about others' investment are driving agents
to invest more. In other words, confidence is founded on common perceptions
of what is going on in the economy. Ignorant agents' (prior) views
will have a disproportionate role in determining confidence. Similarly,
in asset markets with frequent re-trading and random matching, assets
will sometimes pass through the hands of ignorant agents. Their views
will form a focal point around which market expectations will form. 

\subsubsection*{A Subtlety in the Meaning of ``Informed.''}

To interpret and apply our results, it is important to remember that
``prior'' really means ``belief conditional on public information
only.'' (See Section \ref{subsec:Consensus-without-Irreducibility},
where we note that all our analysis is conditional on public information.)
In view of this, we call an agent ``uninformed'' if the beliefs
of that agent are not sensitive to his private information once we
have conditioned on public information. This might not correspond
to other natural senses of ``uninformed,'' so the distinction is
worth keeping in mind. 

\subsubsection*{Least-Informed versus Public.}

We note in closing that this result is very different from the familiar
case of coordinating on something public or ``commonly understood''
in a beauty contest. The less informed agent's information is \emph{not
}public or approximately public. Indeed, in our example, individuals'
signals are conditionally independent given the state. A highly informed
player's signal provides very good information about the \emph{external
state}, but no further information about the \emph{signals} of the
others who are badly informed.

Moreover, in contrast to the standard case of coordinating on a public
signal, our result does not hinge on a qualitative matter of determining
which information is public (something that, actually, is held constant
as we vary the noise rates). It is rather a quantitative matter of
how low the noise rate of the relatively informed players must be
in order for it to ``wash out'' of the (public) consensus expectation.
As discussed in Section \ref{subsec:Generalization-Beyond-Full-Support},
this can depend in a subtle way on priors and the information structure.
In particular, it can happen that the noise of the more informed players
is vanishing compared to the noise of the less informed, and nevertheless
the structure of the smaller noise is decisive for the consensus expectation.
How small the noise must be in order not to matter depends in general
on the network, priors, and information structure, through quantities
that we have described.

\section{Concluding Discussion}

In Appendix \ref{sec:discussion}, we give some detailed discussions
of important assumptions, as well as some extensions. Here we briefly
summarize some of the key points.

\subsubsection*{Joint Connectedness (Section \ref{sec:irred})}

The assumption of joint connectedness was a key maintained assumption
in our results. In this section, we relate it to properties of the
beliefs and the network\textemdash in particular, the connectedness
of the network and the absence of public events (joint connectedness
implies both properties but is not equivalent to their conjunction).
We also discuss what can be done without joint connectedness. This
comes down to the standard analysis of a Markov matrix where not all
states are recurrent.

\subsubsection*{Heterogeneous Self-Weights (Section \ref{subsec:Heterogeneous}) }

In the linear best-response game, we assumed that all agents put a
common weight $\beta$ on others' actions. If this assumption does
not hold, we may reduce to the case where it does hold by changing
the network. In particular, we show how the linear best-response game
with weights $\left(\beta^{1},\ldots,\beta^{|N|}\right)$ and network
$\Gamma$ has the same solution as the game with a common coordination
weight $\widehat{\beta}$ (that depends on $\left(\beta^{1},\ldots,\beta^{|N|}\right)$)
and an alternative network $\widehat{\Gamma}$. The diagonal entries
of the matrix $\widehat{\Gamma}$ capture the variation in self-weights.
This transformation permits the application of our main results to
the case of heterogeneous self-weights. We give interpretations in
terms of both the financial market and the game.

\subsubsection*{Separability and Connection to \citet*{samet1998iteratedA} (Section
\ref{subsec: samet})}

In Section \ref{sec:consensus-and-network}, we showed that\textemdash fixing
the information structure and network\textemdash there are strictly
positive \emph{pseudopriors} $\left(\lambda_{\bm{\pi},\Gamma}^{i}\right)_{i\in N}$
such that $c(y;\bm{\pi},\Gamma)=\sum_{i}e^{i}\mathbf{E}^{\lambda_{\bm{\pi},\Gamma}^{i}}y.$
At the same time, we made the observation\textemdash which here is
explicit in the subscripts of $\lambda^{i}$\textemdash that those
pseudopriors may depend on both the information structure $\bm{\pi}$
and the network $\Gamma$. We say an information structure $\bm{\pi}$
satisfies \emph{separability} if the pseudopriors depend \emph{only}
on the information structure. Section \ref{sec:cpa-signals} shows
that a common prior on signals is sufficient for separability. In
contrast, the assumptions made for the results on contagion of optimism
and tyranny of the least-informed, are not, in general, consistent
with separability. In \citet*{GolubMorris2016} we give a necessary
and sufficient condition for separability, which describes the boundary
between these cases exactly; Section \ref{subsec: samet} sketches
the essential ideas.

Our results in both this paper and \citet*{GolubMorris2016} relate
closely to and build on those of \citet*{samet1998iteratedA}. The
similarity is that, as in his work, limiting properties of higher-order
expectations are shown to depend only on a summary statistic of the
information structure (in our case, the pseudoprior). Section \ref{subsec: samet}
discusses the difference in the results and techniques in detail.

\subsubsection*{Ex Ante and Interim Interpretation (Section \ref{subsec:ex-ante-interim}) }

We take an ex ante perspective in our analysis: At an initial date,
agents have prior beliefs\textemdash and no information\textemdash about
a state of the world. They then receive information and update their
beliefs. We can interpret the results as answering the question: \emph{How
does the consensus expectation change after agents observe their signals?
}Our results give conditions under which: (i) the beliefs \emph{do
not change} (under common priors over signals); (ii) they change to
the most optimistic conceivable beliefs (contagion of optimism); (iii)
they change to the beliefs of the least-informed (tyranny of the least-informed). 

Though we take an ex ante view throughout, consensus expectations
can be seen from a purely interim perspective. Indeed, consensus expectations
depend only on agents' interim beliefs (across all possible types)\textemdash i.e.
on the belief functions $\bm{\pi}$. We discuss how certain main results
would look if we were to stick to a purely interim interpretation.
As in our discussion of separability above, there is a close connection
to the characterization of the common prior assumption in purely interim
terms given by \citet*{samet1998iteratedA}. We highlight both how
our results can be related to his, and also where an ex ante perspective
makes them distinct. While contagion of optimism has purely interim
interpretation, tyranny of the least-informed depends on assumptions
about priors and has no simple interim interpretation.

\subsubsection*{Agent-Specific Random Variables and Incomplete Information about
the Network (Section \ref{heterogeneous-1}) }

Our focus throughout the paper has been on agents' higher-order expectations
of a random variable of common concern, $y$. But an equally interesting
application considers a case where agents have different preferred
actions (which correspond to the different random variables $y^{i}$)
in the absence of coordination motives, and where one's network neighbors
also influence one's choice, with linear best responses assumed \citep*{Ballester2006,cadp15,behm15}.
This case can be embedded readily into our formalism. Indeed, we can
define our $x^{i}(n)$ almost identically to capture this case. This
embodies an equivalence between different priors over the external
states and caring about different random variables\textemdash an equivalence
which does \emph{not }extend to higher-order beliefs, as we explain.
In discussing this connection, we highlight how our results relate
to \citet*{cadp15} and \citet*{behm14}. 

A related point is that there need not be common perceptions or complete
information of the network weights $\gamma^{ij}$. By allowing these
to depend on individuals' types, we can embed incomplete information
about the network into our framework.

\subsection*{Static Higher-Order Expectations, Dynamic Conditional Expectations,
Behavioral Learning, and the DeGroot Model\label{de groot-1}}

We have studied higher-order average expectations of a random variable
in this paper. These higher-order expectations may be interpreted
as being computed at a moment of time. We can call them ``static
higher-order expectations,'' as they are properties of the agents'
static beliefs and higher-order beliefs at that moment. All the iteration
of computing higher-order expectations occurs ``in the agents' minds''
rather than in an interactive dynamic process unfolding over time. 

These static higher-order expectations can be contrasted with agents'
``dynamic conditional expectations'': the beliefs formed via a dynamic
process of updating expectations \emph{after} observing other agents'
conditional expectations up to that point. In this section, we will
use this dichotomy to discuss connections with some important related
literatures.

\citet*{degroot1974reaching} suggested a behavioral model where,
at each stage in a process, each of many agents takes a weighted average
of the beliefs or estimates of his neighbors. He interpreted this
as a heuristic procedure according to which statisticians might average
their own estimates or beliefs with the estimates or beliefs of others
whose opinions they respect, toward the goal of reaching a reasonable
consensus.\footnote{This work grew out of studying aggregation procedures for statistical
estimates. \citet*{Lehrer1981} worked on a related model, seeking
normative foundations for agents' weights in the consensus, based
on the problem of aggregating views in a network of peers. \citet*{Friedkin1999}
studied versions of this model in which each agent persistently weights
a fixed opinion, which can be interpreted as a personal ideal point\textemdash see
Section \ref{heterogeneous-1} for a version of this in our setting.
See \citet*{golub-sadler}, whose Section 3.5.1 we have partly paraphrased
here.} In the DeGroot model, the vector of agents' estimates at stage $n$
is $x(n)=\Gamma^{n}x(0)$, where (as in our model) $\Gamma$ is an
exogenous, fixed stochastic matrix corresponding to the weights agents
assign to various others. Under the classical interpretation, the
DeGroot model is a dynamic process, where agents start out with different
estimates (perhaps based on their private information) and then updating
occurs according to a behavioral rule. \ Economic foundations and
implications of this process have been developed by \citet*{DeMarzo2003},
\citet*{GolubJackson2010}, \citet*{molavi-foundations}, and others.
\ 

Mathematically, the complete-information special case of our static
higher-order expectations model is isomorphic to the classic DeGroot
model, in the sense that equation (\ref{eqn:short-xn}) for updating
the vector of static higher-order expectations,\footnote{Note that under complete information, $B=\Gamma$.}
$x(n)=\Gamma^{n-1}Fy$, looks very much like a DeGroot rule of the
form $x(n)=\Gamma^{n}x(0)$. But it has a different interpretation.
Our agents start out with different priors, captured by $Fy$. In
the dynamic interpretation, $x(2)$ corresponds to taking the weighted
average of neighbors' first-period beliefs. In the static interpretation,
$x(2)$ contains agents' expectation of the average first-order expectations
of others. In this static interpretation, agents' higher-order expectations
are fully Bayesian but based on heterogeneous priors and no asymmetric
information, with weights (i.e., the network $\Gamma$) which are
taken as exogenous. 

Indeed, the general incomplete-information version of our model can
also be related to the DeGroot model. If we draw a parallel where
the types in our model correspond to DeGroot agents, and $x(1)$ is
taken to be the profile of initial estimates, then the ``DeGroot
estimate'' of a given \emph{type} at stage $n$ is the $n^{\text{th}}$-order
iterated average expectation of that type in our model. In this way,
our model can be viewed as an alternative interpretation of DeGroot's
formulas. 

Despite the formal similarity, substantively, the two interpretations
differ very significantly in how they answer a key question in the
DeGroot model literature: How does the network $\Gamma$ affect the
ultimate consensus? Recall that in the DeGroot model, the consensus
is a weighted average of the agents' initial opinions, with the weight
of an agent equal to her eigenvector centrality. (Thus, in DeGroot's
model, if high-centrality agents have high first-order expectations,
the consensus will also be high.) There is an analogous centrality
formula in our setting: Proposition \ref{prop:char}. Despite this,
in our model, under the common prior assumption, there is \emph{no}
interesting dependence of outcomes on $\Gamma$, even when the network
gives some agents very large network centrality: Higher-order average
expectations will always converge to the common prior estimate, independent
of the network. It is only when agents have \emph{heterogeneous} priors
that the network matters in our model. Thus, whereas in the dynamic
learning DeGroot model, the updating implies that centrality always
matters, the additional structure present in our model says that it
matters (to our outcomes) only in specific circumstances, and not
under the common prior assumption.

There is another approach to DeGroot's questions that is different
from his own behavioral model and from our interpretation of his equations
sketched above. That approach is to study standard Bayesian agents
learning dynamically from each other's beliefs, making Bayesian inferences
at each stage. In this case we get a very different updating process.
\citet*{Geanakoplos1982} considered this updating process under the
common prior assumption. Their finding\textemdash in a finite-state
model\textemdash was that posteriors would converge and there would
be common certainty of posteriors in the limit. This model has been
generalized in various directions. For example, \citet*{Parikh1990}
considered the case when one observes posteriors of only some neighbors,
while \citet{nielsen1990common} studied the partial revelation of
posteriors. Recently, \citet*{Rosenberg2009} and \citet*{Mueller-Frank2013}
have explored such models further. Taken together, this literature
provides a fairly rich understanding of dynamically updating conditional
expectations with common priors and asymmetric information on a general
unweighted graph. Note that it contrasts sharply with our analysis;
in the model we have studied in this paper, private information gets
``washed out'' rather than aggregated as we take $n$ to the infinite
limit.

\begin{spacing}{1.1}
\bibliographystyle{ecta}
\bibliography{enc}

\end{spacing}

\pagebreak{}

\appendix

\section{Omitted Proofs\label{sec:Omitted-Proofs}}

\subsection{Proof of Fact \ref{fact:rationalizable}}

\label{subsec:proof-of-rationalizable}To establish (\ref{rationalizable}),
write $R^{i}(k)$ for the set of $i$'s pure strategies surviving
$k$ rounds of iterated deletion of strictly dominated strategies.
By assumption, $R^{i}(k)$ =$R^{i}(0)=[0,M]^{T^{i}}.$ Then using
(\ref{eq:best-response}),
\begin{align*}
R^{i}(1) & =\left\{ s^{i}:(1-\beta)E^{i}y\leq s^{i}\leq(1-\beta)E^{i}y+\beta M\boldsymbol{1}\right\} \\
 & =\left\{ s^{i}:(1-\beta)x^{i}\left(1\right)\leq s^{i}\leq(1-\beta)x^{i}\left(1\right)+\beta M\boldsymbol{1}\right\} 
\end{align*}
For induction, we may assume that for some $k\geq1$, each $R^{i}(k)$
for $i\in N$ has the form 
\[
R^{i}(k)=\left\{ s^{i}:(1-\beta)\left(\sum_{n=1}^{k}\beta^{n-1}x^{i}(n)\right)\leq s^{i}\leq(1-\beta)\sum_{n=1}^{k}\beta^{n-1}x^{i}(n)+\beta^{k}M\boldsymbol{1}\right\} .
\]
We have already established the base case, $k=1$. We will argue that
then

\[
R^{i}(k+1)=\left\{ s^{i}:(1-\beta)\left(\sum_{n=1}^{k+1}\beta^{n-1}x^{i}(n)\right)\leq s^{i}\leq(1-\beta)\left(\sum_{n=1}^{k+1}\beta^{n-1}x^{i}(n)\right)+\beta^{k+1}M\boldsymbol{1}\right\} .
\]
The reason is that if $i$ conjectures a strategy profile $s$ satisfying
\[
(1-\beta)\left(\sum_{n=1}^{k}\beta^{n-1}x^{i}(n)\right)\leq s^{j}
\]
 for each $j\neq i$, then since best responses $\text{BR}^{i}(s)$
are nondecreasing in $s$, the minimum best response $s^{i}$ is obtained
by applying $\text{BR}^{i}$ to the lower bound 
\[
(1-\beta)\left(\sum_{n=1}^{k}\beta^{n-1}x^{i}(n)\right),
\]
 which yields 
\begin{align*}
(1-\beta)E^{i}y+\beta{\displaystyle \sum\limits _{j\neq i}\gamma^{ij}E^{i}(1-\beta)\left(\sum_{n=1}^{k}\beta^{n-1}x^{j}(n)\right)}\\
=(1-\beta)\left(\sum_{n=1}^{k+1}\beta^{n-1}x^{i}(n)\right).
\end{align*}
 The argument for the upper bound is analogous. As $k\rightarrow\infty$,
the lower and upper bounds both converge to the $s_{*}(\beta)$ of
(\ref{rationalizable}). 

\subsection{Existence and Characterization of the Consensus Expectation: Proof
of Proposition \ref{prop:char}}

\textbf{\label{subsec:Existence-and-Characterization-Proof}}Recall
that $p$ is the unique vector in $p\in\Delta(S)$ satisfying $p=pB$;
this vector is uniquely determined and positive by a standard result
for irreducible Markov chains. Write 
\begin{equation}
x(\beta)=(1-\beta)\sum_{n=0}^{\infty}\beta^{n}B^{n}z.\label{eq:xr-summation}
\end{equation}
We will show that for any $z\in\mathbb{R}^{S}$, we have 
\begin{equation}
\lim_{\beta\uparrow1}x(\beta)=pz\mathbf{1}.\label{eq:limit-x-goal}
\end{equation}

Note that by the Neumann series, which can be used since the spectral
radius of $\beta B$ is $\beta<1$, we have $\sum_{n=0}^{\infty}(\beta B)^{n}=(I-\beta B)^{-1}$,
where $I$ denotes the identity matrix of appropriate size; in particular,
$I-\beta B$ is invertible. So $x(\beta)=(1-\beta)(I-\beta B)^{-1}z,$
or, equivalently,
\begin{equation}
(I-\beta B)x(\beta)=(1-\beta)z.\label{eq:x-rearranged}
\end{equation}
The formula (\ref{eq:xr-summation}) says that $x(\beta)$ is an average,
because the weights $(1-\beta)\beta^{n}$ sum to 1, of the vectors
$B^{n}z$. Because $B^{n}$ is a Markov matrix, no entry of $B^{n}z$
can exceed the largest value of $z$ in absolute value. So the same
is true of $x(\beta)$, and therefore all the $x(\beta)$ lie in a
compact set. 

Consider a sequence $\beta_{k}\uparrow1$. By what we have said, the
sequence $(x(\beta_{k}))_{k}$ lies inside a compact set. By a standard
fact about compact sets, such a sequence converges, and has the limit
$pz\mathbf{1}$, if and only if every convergent subsequence of it
converges to $pz\mathbf{1}$. So consider a convergent subsequence,
$(x(\beta_{\kappa}))_{\kappa}$, and let $x$ denote its limit. We
will show that $x=pz\bm{1}$, which will conclude the proof of (\ref{eq:limit-x-goal}).

By taking $\beta\uparrow1$ in (\ref{eq:x-rearranged}), we see that
$x$ satisfies $x=Bx$, which, given that our matrix $B$ is irreducible,
means that $x=a\mathbf{1}$ for some constant $a$. It remains only
to prove that $a=pz$. Premultiplying (\ref{eq:x-rearranged}) by
$p$ gives $(1-\beta_{\kappa})px(\beta_{\kappa})=(1-\beta_{\kappa})pz$.
Canceling $(1-\beta_{\kappa})$, we get $px(\beta_{\kappa})=pz$.
Letting $\kappa\to\infty$ and recalling that $x$ is defined as the
limit of the subsequence yields $px=pz$. When we plug in $x=a\mathbf{1}$\textemdash the
statement that $x$ is a constant vector\textemdash we find that $ap\mathbf{1}=pz$.
Since $p$ is a probability vector, we have $p\mathbf{1}=\mathbf{1}$,
and so we conclude that $a=pz$. 

\subsection{Proof of Lemma \ref{lem:markov-optimism}}

If $W_{1}$ is drawn from the ergodic distribution $p$, the distributions
of $W_{1}$ and $W_{2}$ are the same, and so the expected difference
between $f(W_{2})$ and $f(W_{1})$ is $0$: 
\begin{equation}
\mathbb{E}_{W_{1}\sim p}[f(W_{2})-f(W_{1})]=0.\label{eq:zero-difference}
\end{equation}
On the other hand, using hypotheses (1) and (2) in the second line
below, we have
\begin{eqnarray*}
\mathbb{E}_{W_{1}\sim p}[f(W_{2})-f(W_{1})] & = & \sum_{s:f(s)<\overline{f}}p(s)\mathbb{E}_{W_{1}=s}[f(W_{2})-f(s)]+\sum_{s:f(s)\geq\overline{f}}p(s)\mathbb{E}_{W_{1}=s}[f(W_{2})-f(s)]\\
 & \geq & \delta p(s:f(s)<\overline{f})-\varepsilon p(s:f(s)\geq\overline{f}).
\end{eqnarray*}
Combining this result with (\ref{eq:zero-difference}) and using the
shorthand $\chi=p(s:f(s)\geq\overline{f})$, we deduce $0\geq\delta(1-\chi)-\varepsilon\chi$,
from which the lower bound on $\chi$ claimed in the proposition follows.

\subsection{Proof for Claims in Section \ref{subsec:optimism-tightness} about
Tightness Result}

\label{subsec:appendix-tightness}To demonstrate the claim made in
Section \ref{subsec:optimism-tightness}, first note that the chain
satisfies the assumptions of Lemma \ref{lem:markov-optimism} with
$\overline{f}=m$. Let $S_{k}$ be the set of states $\left\{ t_{k}^{i}:i\in\left\{ 1,2\right\} \right\} $.
The stationary mass entering $S_{m}$ has to be equal to the mass
exiting it. Transitions to $S_{m}$ come only from $S_{m-1}$. Finally,
the absorbing states are $S_{m-1}\cup S_{m}$. Combining these facts:
\[
p(S_{m})\varepsilon=p(S_{m-1})\delta=[1-p(S_{m})]\delta,
\]
so that $p(S_{m})=1/(1+\varepsilon/\delta)$. A slight perturbation
of the chain will result in very nearly the same bound for an irreducible
chain. Note that we can generate such an example for as many agents
as we want, and as many types per agent (so tightness is established
for all ``sizes'' of the setting).

\subsection{Proofs of Results on Tyranny of the Least-Informed}

\label{sec:tyranny-proofs}The key lemma behind our proof of Proposition
\ref{prop:ignorant-player} is:
\begin{lem}
\label{lem:p-inequality}Under the hypotheses of Proposition \ref{prop:ignorant-player},
\[
\left|\frac{p(s)-\widehat{p}(s)}{\widehat{p}(s)}\right|\leq\frac{4|\Theta||S|^{2}}{(\gamma_{\min}\rho_{\min})^{2}}\cdot\frac{\varepsilon}{\delta}.
\]
\end{lem}
\begin{proof}
The proof relies on Theorem 2.1 of \citet*{Cho-Meyer}, which says
that, for any $s\in S$, 
\begin{equation}
\left|\frac{p(s)-\widehat{p}(s)}{\widehat{p}(s)}\right|\leq\frac{1}{2}\left\Vert B-\widehat{B}\right\Vert _{\infty}\max_{z\neq z'}M_{\widehat{B}}(z,z'),\label{eq:cho-meyer}
\end{equation}
where $M_{\widehat{B}}(z,z')$ is the mean first passage time\footnote{Consider a Markov chain making transitions according to $\widehat{B}$.
The mean first-passage time from $z$ to $z'$ in $\widehat{B}$ is
denoted by $M_{\widehat{B}}(z,z')$ and defined to be the expected
number of steps that the chain started at $z$ takes up to its first
visit to $z'$ (inclusive).} in $\widehat{B}$ to $z'$ starting at $z$; the norm is the maximum
absolute row sum. Two key technical lemmas, stated in Section \ref{subsec:tech-lem}
below, allow us to bound the right-hand side. Using Lemma \ref{lem:Bij-bound}
(summing the upper bounds on absolute differences across any row and
taking the maximum over all rows $i$):
\[
\left\Vert B-\widehat{B}\right\Vert _{\infty}\leq|S|\cdot\frac{4|\Theta||S|\varepsilon}{\min_{i\neq1}\rho_{\min}^{i}}.
\]
To finish bounding the right-hand side of (\ref{eq:cho-meyer}), it
remains to bound $\max_{z\neq z'}M_{\widehat{B}}(z,z')$. Lemma \ref{lem:first-passage-bound}
does exactly this, giving
\[
\max_{z\neq z'}M_{\widehat{B}}(z,z')\leq\frac{2}{\delta\rho_{\min}^{1}\gamma_{\min}^{2}}.
\]
Recall that $\gamma_{\min}$ is the minimum off-diagonal entry of
$\Gamma$\textemdash by assumption a positive number. Combining the
two inequalities gives the claimed bound.
\end{proof}
Now we can show how this result implies Proposition \ref{prop:ignorant-player}. 

The first step is to show that the consensus expectation under the
hatted information structure is equal to the first agent's prior expectation:
\[
c(y;B_{\widehat{\bm{\pi}}},F_{\widehat{\bm{\pi}}})=\mathbf{E}^{\rho^{1}}[y].
\]
The key to this is to establish that the information structure $(\widehat{\pi}^{i})_{i\in N}$
is consistent with a common prior over signals. Indeed, we will show
that agent $1$'s prior can be taken to be this common prior. Let
$\widehat{\mu}^{1}\in\Delta(T^{1})$ be the prior on $T^{1}$ induced
by $\rho^{1}$, and let 
\[
\widehat{\mu}^{i}(t^{i})=\sum_{t^{1}\in T^{1}}\widehat{\pi}^{1}(t^{i}\mid t^{1})\widehat{\mu}^{1}(t^{1}).
\]
For agents $i\neq1$, the interim beliefs $\widehat{\pi}^{i}(\cdot\mid t^{i})$
are compatible with their respective priors $\widehat{\mu}^{i}$ trivially,
because the interim beliefs place probability $0$ or $1$ on any
state, and are compatible with \emph{any }prior\textemdash Bayes'
rule implies no restrictions. Moreover, with this profile $(\widehat{\mu}^{i})_{i\in N}$,
the information structure $(\widehat{\pi}^{i})_{i\in N}$ is consistent
with a common prior over signals. Now note that the prior over $\Theta$
corresponding to any $\widehat{\mu}^{i}$ is $\rho^{1}$. By Proposition
\ref{prop:cpa-signals}, the consensus expectation $c(y;B_{\widehat{\bm{\pi}}},F_{\widehat{\bm{\pi}}})$
is the common prior expectation of $y$, namely $\mathbf{E}^{\rho^{1}}[y]$.

The second step is to bound the distance between $c(y;B_{\widehat{\bm{\pi}}},F_{\widehat{\bm{\pi}}})$,
which we have computed, and $c(y;B_{\bm{\pi}},F_{\bm{\pi}})$, which
we would like to characterize. It is here that Lemma \ref{lem:p-inequality}
is relevant:

\begin{align*}
\left|c(y;B_{\bm{\pi}},F_{\bm{\pi}})-c(y;B_{\widehat{\bm{\pi}}},F_{\widehat{\bm{\pi}}})\right| & =\left|\sum_{s\in S}[p(s)-\widehat{p}(s)]E^{i}[y\mid s]\right|\\
 & =\left|\sum_{s\in S}\frac{p(s)-\widehat{p}(s)}{\widehat{p}(s)}\widehat{p}(s)E^{i}[y\mid s]\right| & \text{multiply and divide by \ensuremath{\widehat{p}(s)}}\\
 & \leq\sum_{s\in S}\left|\frac{p(s)-\widehat{p}(s)}{\widehat{p}(s)}\right|\widehat{p}(s)\left|E^{i}[y\mid s]\right| & \text{triangle inequality}\\
 & \leq\frac{4|\Theta||S|^{2}}{(\gamma_{\min}\rho_{\min})^{2}}\cdot\frac{\varepsilon}{\delta}\sum_{s\in S}\widehat{p}(s)\left|E^{i}[y\mid s]\right| & \text{Lemma \ref{lem:p-inequality}}\\
 & \leq\frac{4|\Theta||S|^{2}}{(\gamma_{\min}\rho_{\min})^{2}}\cdot y_{\max}\cdot\frac{\varepsilon}{\delta}. & \text{definition of }y_{\text{max}}
\end{align*}
This completes the proof of the proposition, except for the technical
lemmas, which are the subject of the next section.

\subsubsection{Statements of Technical Lemmas}

\label{subsec:tech-lem}The proof of Lemma \ref{lem:p-inequality}
used two key bounds. We state both here, and give proofs in Appendix
\ref{sec:online-proofs}.

The first result, which was used to bound $\Vert B-\widehat{B}\Vert_{\infty}$,
converts hypotheses about the signal structures $(\eta^{i})_{i\in N}$
into statements about the agents' interim beliefs (recall that the
entries of $B$ are products of network weights from $\Gamma$ and
interim beliefs):
\begin{lem}
\label{lem:Bij-bound}For any $t^{i},t^{j}\in S$ with $j\neq i$,
we have
\[
\left|\pi^{i}(t^{j}\mid t^{i})-\widehat{\pi}^{i}(t^{j}\mid t^{i})\right|\leq\frac{4|\Theta||S|\varepsilon}{\rho_{\min}^{i}}.
\]
\end{lem}
This follows from Bayes' rule, but the exact statement requires a
good deal of calculation. The core idea is that $\widehat{\eta}$
is obtained by changing the probabilities in $\eta$ only slightly.
Given full support priors, each $\pi^{i}(t^{j}\mid t^{i})$ is continuous
in $\eta^{i}(t^{i}\mid\theta)$, so it is natural that the two should
be close; our calculation simply gives a quantitative version of this
statement. \medskip{}

We also used a bound on mean first-passage times in $\widehat{B}$: 
\begin{lem}
\label{lem:first-passage-bound} For any two states $z,z'\in S$,
\[
M_{\widehat{B}}(z,z')\leq\frac{2}{\delta\rho_{\min}^{1}\gamma_{\min}^{2}}.
\]
\end{lem}
The key idea here is that, as a consequence of agent $1$ having noisy
information, the subjective probability agent $1$ puts on any type
of any other agent is reasonably high: The lower bound is $\rho_{\min}^{1}\delta$,
as we establish in the proof. Thus the corresponding weights in $\widehat{B}$
are lower-bounded by $\delta\rho_{\min}^{1}\gamma_{\min}$, once we
take into account the network part of the weight. The other agents'
types have perfect information, so each of them has an edge of weight
at least $\gamma_{\text{min}}$ to a type of agent $1$. Thus the
Markov chain is well-interconnected by agent $1$'s types: Starting
from any state, one gets to agent $1$'s types immediately, and then
to any other given state in $S$ with substantial probability, so
the chain cannot take too long to visit that state (by a standard
bound on geometric random variables).

The proofs of the technical lemmas appear in Appendix \ref{sec:online-proofs}.

\pagebreak{}

\section{For Online Publication: Proofs of Technical Lemmas\label{sec:online-proofs}}

\subsection{\label{sec:proof-B-close}Proof of Lemma \ref{lem:Bij-bound}}

The proof relies on the following fact about prior probabilities of
signals.
\begin{fact}
\label{fact:prior-lower-bound}For any $i\neq1$ and any $t^{i}$,
\[
\mu^{i}(t^{i})=\sum_{\theta'\in\Theta}\eta^{i}(t^{i}\mid\theta')\rho^{i}(\theta')\geq(1-\varepsilon)\rho_{\min}^{i}.
\]
\end{fact}
This bound holds because $\eta^{i}$ is assumed to be at most $\varepsilon$-noisy,
and so there must be some $\theta_{t_{i}}$ such that $\eta^{i}(t^{i}\mid\theta_{t_{i}})\geq1-\varepsilon$.
\bigskip{}

The first step of the proof of Lemma \textbf{\ref{lem:Bij-bound}}
is to write the probabilities in question via sums over states $\theta$.
For any $t^{i},t^{j}$ with $j\neq i$, we have
\[
\pi^{i}(t^{j}\mid t^{i})=\sum_{\theta\in\Theta}\eta^{j}(t^{j}\mid\theta)\pi^{i}(\theta\mid t^{i})
\]
 Define $\widehat{\pi}^{i}(t^{j}\mid t^{i})$ analogously, replacing
$\pi^{i}$ by $\widehat{\pi}^{i}$ and $\eta^{i}$ by $\widehat{\eta}^{i}$.
Let 
\[
H^{j}(t^{j}\mid\theta)=\left|\eta^{j}(t^{j}\mid\theta)-\widehat{\eta}^{j}(t^{j}\mid\theta)\right|
\]
and
\[
\Delta^{i}(\theta\mid t^{i})=\left|\pi^{i}(\theta\mid t^{i})-\widehat{\pi}^{i}(\theta\mid t^{i})\right|.
\]

Now note that by the triangle inequality,
\begin{equation}
\left|\pi^{i}(t^{j}\mid t^{i})-\widehat{\pi}^{i}(t^{j}\mid t^{i})\right|\leq\sum_{\theta\in\Theta}\left[\Delta^{i}(\theta\mid t^{i})+H^{j}(t^{j}\mid\theta)+\Delta^{i}(\theta\mid t^{i})H^{j}(t^{j}\mid\theta)\right].\label{eq:difference-as-sum}
\end{equation}

Having written the difference we are studying in this way, we will
bound it piece by piece. If $j\neq1$, by definition of ``at most
$\varepsilon$-noisy,'' we have that $|H^{j}(t^{j}\mid\theta)|\leq\varepsilon$.
If $j=1$, then $H^{j}(t^{j}\mid\theta)$ is identically zero. Also,
note that $|\Delta^{i}(\theta\mid t^{i})|\leq1$. So in all cases,
we can bound the last two terms in the brackets by $2\varepsilon$. 

Now, we turn to $\Delta^{i}(\theta\mid t^{i})$. If $i=1$, then $\Delta^{i}(\theta\mid t^{i})=0$,
because 1's signals are the same in both the original information
structure $\bm{\pi}$ and the new one $\widehat{\bm{\pi}}$. 

So assume $i\neq1$; we will show that $\Delta^{i}(\theta\mid t^{i})\leq(|S|-1)\frac{\varepsilon}{(1-\varepsilon)\rho_{\min}^{i}}$,
and this will allow us to complete the proof. Let $\theta_{t_{i}}$
be such that $\eta^{i}(t^{i}\mid\theta_{t_{i}})\geq1-\varepsilon$,
which is guaranteed to exist by the definition of ``at most $\varepsilon$-noisy.''
We will bound $\Delta^{i}(\theta\mid t^{i}),$ considering the cases
$\theta\neq\theta_{t_{i}}$ and $\theta=\theta_{t^{i}}$ separately.
If $\theta\neq\theta_{t_{i}}$, then by Bayes' rule, 
\begin{align*}
\pi^{i}(\theta\mid t^{i}) & =\frac{\eta^{i}(t^{i}\mid\theta)\rho^{i}(\theta)}{\mu^{i}(t^{i})}\\
 & \leq\frac{\eta^{i}(t^{i}\mid\theta)\rho^{i}(\theta)}{(1-\varepsilon)\rho_{\min}^{i}} & \text{by Fact \ref{fact:prior-lower-bound} }\\
 & \leq\frac{\varepsilon}{(1-\varepsilon)\rho_{\min}^{i}} & \text{\text{by definition of at least \ensuremath{\varepsilon}-nosiy}. }
\end{align*}
Since $\widehat{\pi}^{i}(\theta\mid t^{i})=0$, it follows that 
\begin{equation}
\Delta^{i}(\theta\mid t^{i})\leq\frac{\varepsilon}{(1-\varepsilon)\rho_{\min}^{i}}\label{eq:Delta-bound-1}
\end{equation}
 By the law of total probability,
\[
\pi^{i}(\theta_{t^{i}}\mid t^{i})\geq1-(|S|-1)\frac{\varepsilon}{(1-\varepsilon)\rho_{\min}^{i}}.
\]
Since $\widehat{\pi}^{i}(\theta_{t_{i}}\mid t^{i})=1,$ it follows
that 
\begin{equation}
\Delta^{i}(\theta_{t^{i}}\mid t^{i})\leq(|S|-1)\frac{\varepsilon}{(1-\varepsilon)\rho_{\min}^{i}}.\label{eq:Delta-bound-2}
\end{equation}
This is the looser of the two bounds (\ref{eq:Delta-bound-1}) and
(\ref{eq:Delta-bound-2}), so we can say in general that 
\begin{equation}
\Delta^{i}(\theta_{t^{i}}\mid t^{i})\leq(|S|-1)\frac{\varepsilon}{(1-\varepsilon)\rho_{\min}^{i}}.\label{eq:Delta-bound-3}
\end{equation}

Putting everything together, it follows that

\begin{align*}
\left|\pi^{i}(t^{j}\mid t^{i})-\widehat{\pi}^{i}(t^{j}\mid t^{i})\right| & \leq\sum_{\theta\in\Theta}\left[\Delta^{i}(\theta\mid t^{i})+2\varepsilon\right] & \text{by }(\ref{eq:difference-as-sum})\\
 & \leq\sum_{\theta\in\Theta}\left[(|S|-1)\cdot\frac{\varepsilon}{(1-\varepsilon)\rho_{\min}^{i}}+2\varepsilon\right]\\
 & \leq|\Theta|\left((|S|-1)\frac{\varepsilon}{(1-\varepsilon)\rho_{\min}^{i}}+2\varepsilon\right)\\
 & \leq|\Theta|\varepsilon\left((|S|-1)\frac{2}{\rho_{\min}^{i}}+2\right) & \text{using }1-\varepsilon\geq\frac{1}{2}\\
 & \leq|\Theta|\frac{\varepsilon}{\rho_{\min}^{i}}\left(2(|S|-1)+2\right)
\end{align*}
The claimed bound follows after noting and that $2|S|\geq2(|S|-1)+2$
because $|S|\geq2$.

\subsection{Proof of Lemma \ref{lem:first-passage-bound}}

The proof requires the following fact:
\begin{fact}
\label{fact:ignorant-bound}For any $t^{1}\in T^{1}$ and $t^{j}\in T^{j}$
with $j\neq1$ we have:
\[
\widehat{\pi}^{1}(t^{j}\mid t^{1})=\frac{\sum_{\theta\in\Theta}\rho^{1}(\theta)\widehat{\eta}^{j}(t^{j}\mid\theta)}{\sum_{\theta\in\Theta}\rho^{1}(\theta)\eta^{1}(t^{1}\mid\theta)}\geq\delta\rho_{\min}^{1}.
\]
\end{fact}
To establish this fact, we note that there is some $\theta_{t^{j}}$
such that $\widehat{\eta}^{j}(t^{j}\mid\theta_{t^{j}})=1$, and the
denominator is at most $1$ since it is the prior probability of the
signal $t^{1}$ under the information structure associated with $\widehat{\eta}^{1}$.

\bigskip{}

Now we prove\textbf{ }Lemma \textbf{\ref{lem:first-passage-bound}}.
Let $(\widehat{W}_{n})_{n}$ be a stochastic process corresponding
to the Markov matrix $\widehat{B}$. Defining the function $\iota:S\to N$
by $\iota(t^{i})=i$, we have a coupling between the chain $(\widehat{W}_{n})_{n}$
and a chain on $N$, the set of agents, with transition matrix $\Gamma$. 

\subsubsection*{\emph{Case 1: $z'\protect\notin T^{1}$}}

Let us analyze the first passage time to some $z'\notin T^{1}$. Starting
from any $z\in S$, the mean first passage time of the process $(\iota(\widehat{W}_{n}))_{n}$
to $1$ (the state corresponding to agent $1$) is at most $1/\gamma_{\min}$.
Then every time the process visits a state in $T^{1}$, it has probability
at least $\delta\rho_{\min}^{1}\gamma_{\min}$ of visiting $z'$,
by Fact \ref{fact:ignorant-bound}. Conditional on not visiting it
at this time, we wait on average $1/\gamma_{\min}$ steps for the
process to return to a state in $T^{1}$ and have another $\delta\rho_{\min}^{1}\gamma_{\min}$
chance at visiting $z'$. Thus, using the formula for the expectation
of a geometric random variable, we have
\[
M_{\widehat{B}}(z,z')\leq\frac{1}{\delta\rho_{\min}^{1}\gamma_{\min}^{2}}\text{\quad}\text{whenever }z'\notin T^{i}.
\]

\subsubsection*{\emph{Case 2: $z'\in T^{1}$}}

Let $z'=t^{1}.$ If $z=t^{j}\notin T^{1}$, then there is a $\theta_{t^{j}}$
such that $\widehat{\pi}^{j}(\theta_{t^{j}}\mid t^{j})=1$. Then 
\[
\widehat{\pi}^{j}(t^{1}\mid t^{j})=\eta^{1}(t^{1}\mid\theta_{t^{j}}),
\]
which is at least $\delta$ by the definition of ``at least $\delta$-noisy.''
Thus, every time the process $(\widehat{W}_{n})_{n}$ visits any state
in $S\setminus T^{1}$, it has probability at least $\delta$ of visiting
$z'$. If $z\in T^{1}$, then the process surely visits the set $S\setminus T^{1}$
one step later. Thus the process takes at most two steps to be in
a position where it has probability $\delta$ of visiting $z$. By
the same reasoning discussed above about a geometric random variable,
we conclude that 
\[
M_{\widehat{B}}(z,z')\leq\frac{2}{\delta}.
\]

\section{For Online Publication: Discussion of Assumptions and Variants of
our Results\label{sec:discussion}}

We now discuss robustness and extensions of our results (Sections
\ref{subsec:Irreducibility} and \ref{subsec:Heterogeneous}), as
well as their context and broader implications (Sections \ref{subsec: samet}
through \ref{de groot-1}). More technical issues are postponed to
Appendix \ref{sec:Additional-Discussion}. 

\subsection{Joint Connectedness\label{sec:irred}}

\label{subsec:Irreducibility}An assumption maintained throughout
was a joint connectedness condition (recall Section \ref{sec:joint-connected}),
which amounts to the interaction structure $B$ being irreducible.
In the present subsection, we no longer treat this condition as a
maintained assumption, and examine its content and what can be said
without it. Proposition \ref{prop:irred-1} reviews a characterization
of the irreducibility condition: It is equivalent to the agent-type
vector $p$ being strictly positive. We then relate the condition
to properties of the primitives $\Gamma$ and $\bm{\bm{\pi}}$. Finally,
we discuss results that hold under weakenings of the assumption. 

\subsubsection{Relations to Beliefs and the Network\label{sec:irreducibility-primitives}}

Some key properties of the network and beliefs will feature in our
characterization of irreducibility. A network $\Gamma$ is \emph{complete}
if $\gamma^{ij}>0$ for all $i$ and $j$. \ Beliefs $\bm{\pi}$
have \emph{full support marginals} if $\pi^{i}\left(t^{j}\mid t^{i}\right)>0$
for all agents $i$ and $j$, and all signals $t^{i}\in T^{i}$, $t^{j}\in T^{j}$.
Event $G\subseteq T$ is a \emph{product event} if $G=\prod_{i\in N}G_{i}$,
where $G^{i}\subseteq T^{i}$ for each $i$. Say that a product event
$G=\prod_{i\in N}G_{i}$, is a \emph{public} or \emph{closed} event
(under beliefs $\bm{\pi}$) if, for each agent $i$ and each signal
$t^{i}\in G^{i}$, the following implication holds for any $t^{-i}\in T^{-i}$:
\[
\pi^{i}\left(t^{-i}\mid t^{i}\right)>0\hspace{0.2in}\implies\hspace{0.2in}(t^{i},t^{-i})\in G\text{.}
\]
A public or closed event is one that, when it occurs, is common certainty
among all the agents: For any observed signal, no probability is assigned
any signal outside the event. \emph{Beliefs $\bm{\pi}$ are connected}
if $\varnothing$ and $T$ are the only public events. This corresponds
to the notion of no (nontrivial) common certainty: Every nontrivial
product event has a connection (via beliefs placed by some agent)
to states outside itself. A subset of agents $J\subseteq N$ is \emph{closed}
if $i\in J$ and $\gamma^{ij}>0$ implies $j\in J$. \ A \emph{network
$\Gamma$ is connected} if $\varnothing$ and $N$ are the only closed
subsets of the agent set $N$. Recall that a network is a \emph{complete
}if $\gamma^{ij}>0$ for all $i$ and $j$. 

The properties mentioned so far are restrictions on either the beliefs
or the network, but not both. The property of joint connectedness
from Section \ref{sec:joint-connected} is a joint restriction on
both. The following result relates the two sorts of conditions. 
\begin{prop}
\label{prop:irred-1} The matrix $B$ is irreducible if and only if
beliefs and the network are jointly connected. Necessary conditions
for this are: 

\begin{enumerate}
\item Beliefs are connected. 
\item The network is connected. 
\end{enumerate}
Sufficient conditions for this are: 

\begin{enumerate}
\item The network is complete and beliefs are connected. 
\item The network is connected and beliefs have full support marginals. 
\end{enumerate}
\end{prop}
\begin{proof}
The ``if and only if'' part is just a rewriting of the statement
that there are no nonempty, proper closed communicating classes in
the Markov process corresponding to $B$. The two sufficient conditions
are strengthenings of this property. 
\end{proof}
The following example illustrates that requiring a connected network
and connected beliefs separately is not sufficient for irreducibility.
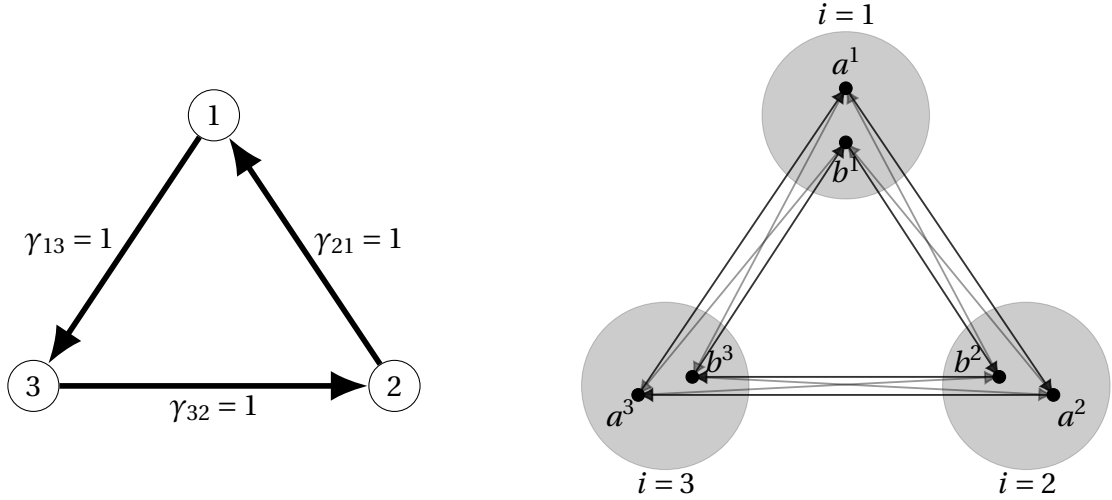
\begin{figure}
\begin{centering}
\begin{tikzpicture}[scale=1.2]   
\coordinate [label=left:{$\gamma_{13}=1$}] (zero) at (2,1.6); 
\coordinate [label=right:{$\gamma_{21}=1$}] (zero) at (4,1.6); 
\coordinate [label=below:{$\gamma_{32}=1$}] (zero) at (3,0); 
\node[circle, draw] (1) at (3,3) {$1$};  
\node[circle, draw] (3) at (1,0) {$3$};  
\node[circle, draw] (2) at (5,0) {$2$}; 
\draw[line width=.75mm,-{Latex[length=5mm, width=4mm]}] (1) to (3); 
\draw[line width=.75mm,-{Latex[length=5mm, width=4mm]}] (2) to (1); 
\draw[line width=.75mm,-{Latex[length=5mm, width=4mm]}] (3) to (2);


\begin{scope}[transparency group, opacity=0.2] \draw[fill=black] (8,0) circle (5ex); 
\end{scope} 
\begin{scope}[transparency group, opacity=0.2] \draw[fill=black] (12,0) circle (5ex); \end{scope} 
\begin{scope}[transparency group, opacity=0.2] \draw[fill=black] (10,3) circle (5ex); \end{scope}
\coordinate [label=below:{$i=3$}] (a) at (8,-.85);   
\coordinate [label=below:{$i=2$}] (a) at (12,-.85);  
\coordinate [label=below:{$i=1$}] (a) at (10,4.35);  
\filldraw[black] (7.7,-.1) circle (2pt) ; 
\filldraw[black] (8.3,.1) circle (2pt) ; 
\coordinate [label=center:{$a^3$}] (a) at (7.5,-.3);   
\coordinate [label=center:{$b^3$}] (a) at (8.6,.3);                       
\filldraw[black] (12.3,-.1) circle (2pt) ; 
\filldraw[black] (11.7,.1) circle (2pt) ; 
\coordinate [label=center:{$a^2$}] (a) at (12.5,-.3);   
\coordinate [label=center:{$b^2$}] (a) at (11.35,.3);  
\filldraw[black] (10,2.7) circle (2pt) ; 
\filldraw[black] (10,3.3) circle (2pt) ; 
\coordinate [label=center:{$b^1$}] (a) at (10,2.4);   
\coordinate [label=center:{$a^1$}] (a) at (10,3.6);  
\begin{scope}[transparency group, opacity=0.8] 
\draw[line width=.25mm,-{Latex[length=2mm, width=2mm]}]  (7.7, -.1) --(10, 3.3); \end{scope} 
\begin{scope}[transparency group, opacity=0.8] \draw[line width=.25mm,-{Latex[length=2mm, width=2mm]}] (10, 3.3)--(12.3,-.1) ; \end{scope} 
\begin{scope}[transparency group, opacity=0.8] \draw[line width=.25mm,-{Latex[length=2mm, width=2mm]}] (12.3,-.1)--(7.7, -.1)  ; \end{scope}
\begin{scope}[transparency group, opacity=0.8] \draw[line width=.25mm,-{Latex[length=2mm, width=2mm]}]  (8.3, .1) --(10, 2.7); \end{scope} 
\begin{scope}[transparency group, opacity=0.8] \draw[line width=.25mm,-{Latex[length=2mm, width=2mm]}] (10, 2.7)--(11.7,.1) ; \end{scope} 
\begin{scope}[transparency group, opacity=0.8] \draw[line width=.25mm,-{Latex[length=2mm, width=2mm]}] (11.7,.1)--(8.3, .1)  ; \end{scope}
\begin{scope}[transparency group, opacity=0.4] \draw[line width=.25mm,-{Latex[length=2mm, width=2mm]}]  (10, 2.7)--(7.7, -.1); \end{scope} 
\begin{scope}[transparency group, opacity=0.4] \draw[line width=.25mm,-{Latex[length=2mm, width=2mm]}]  (10, 3.3)--(8.3, .1); \end{scope}
\begin{scope}[transparency group, opacity=0.4] \draw[line width=.25mm,-{Latex[length=2mm, width=2mm]}]  (7.7, -.1)--(11.7,.1) ; \end{scope} \begin{scope}[transparency group, opacity=0.4] \draw[line width=.25mm,-{Latex[length=2mm, width=2mm]}]  (8.3, .1)--(12.3, -.1); \end{scope}
\begin{scope}[transparency group, opacity=0.4] \draw[line width=.25mm,-{Latex[length=2mm, width=2mm]}]  (11.7,.1)--(10, 3.3) ; \end{scope} \begin{scope}[transparency group, opacity=0.4] \draw[line width=.25mm,-{Latex[length=2mm, width=2mm]}]  (12.3, -.1)--(10, 2.7); \end{scope} 
\end{tikzpicture}
\par\end{centering}
\caption{An example illustrating that imposing connectedness of the network
and of beliefs is not sufficient to ensure joint connectedness, i.e.
irreducibility of the interaction structure $B$.}
\label{fig:counterexample-joint-connectedness}
\end{figure}

\begin{example*}
\label{ex:irreducibility-conditions-1} Suppose that there are three
agents and each agent observes one of two signals, so that $T^{i}=\{a^{i},b^{i}\}$.
The network is given by a cycle, 

\[
\Gamma=\left(\begin{array}{ccc}
0 & 1 & 0\\
0 & 0 & 1\\
1 & 0 & 0
\end{array}\right),
\]
and the information structure is such that each agent $i$ is sure
that agent $i+1$ has observed the same signal and that agent $i-1$
has observed a different signal (under mod $3$ arithmetic for agent
indices). Figure \ref{fig:counterexample-joint-connectedness} illustrates
the network and the information structure. In this example, the network
is connected, and the beliefs are connected; nevertheless, $B$ is
not irreducible, because beliefs and the network are not jointly connected.
The subset $\left\{ a^{1},a^{2},a^{3}\right\} $ places no weight
in $B$ on its complement.

\medskip{}
\end{example*}
The conditions we have discussed have placed restrictions on beliefs
directly, rather than on observable consequences. In Appendix \ref{sec:no-trade},
we give a ``no trade'' behavioral characterization of irreducibility.

\subsubsection{Consensus without Irreducibility\label{subsec:Consensus-without-Irreducibility}}

All our results have analogues when irreducibility fails. As we saw
in Section \ref{sec:consensus-expectation}, what matters for the
limit of $x(n;y)$ as $n\to\infty$ is the behavior of $B^{n}$. This
can be characterized quite generally based on the graph described
in Section \ref{sec:consensus-and-network}. The general result can
be found in many textbooks \citep*[e.g.][Section 8.4]{Meyer2000},
and we summarize it informally. First, consider the set $S_{A}$,
defined as the set of absorbing states in $S$ according to the transition
matrix $B$. For such $t^{i}$, the analysis of $x(n;y)$ can proceed
exactly as in Section \ref{sec:consensus-expectation}, restricting
$B$ to the maximal strongly connected component containing $t^{i}$. 

The simplest case is when $S$ can be partitioned into several strongly
connected components (so $S=S_{A}$) with $B$ having no edges between
these components. This occurs, for instance, if there are exactly
two public (product) events. Then the analysis can be done on each
of these components separately. That is, the analysis can be done
\emph{conditional on public information}. More generally, when there
are public events, our assertion that the consensus expectation is
nonrandom (recall Section \ref{sec:consensus-expectation}) really
means that it is nonrandom conditional on the public event that has
occurred (and which, by definition of its being public, is common
knowledge).

Now suppose there \emph{are }some nonabsorbing states. For each non-absorbing
state $t^{i}\notin S_{A}$, the corresponding row of $B^{\infty}$
is a distribution that allocates mass (in a particular way) across
the set $S_{A}$ of absorbing states. 

An important case is relevant to several of our discussions. When
$S_{A}$ consists of exactly one strongly connected component (though
it may be a strict subset of $S$), we can refine our statements above
to obtain the following generalization of Proposition \ref{prop:char}:
\begin{prop}
\label{prop:char-1}If $S_{A}$ has exactly one strongly connected
component, the consensus expectation exists and
\begin{equation}
c(y;\bm{\pi},\Gamma)=\sum_{t^{i}\in S_{A}}p(t^{i})E^{i}[y\mid t^{i}],\label{eq:c-in-terms-of-p-1}
\end{equation}
where $p\in\Delta(S_{A})$, called the vector of \emph{agent-type
weights}, is the stationary distribution of $B_{S_{A}}$ ($B$ restricted
to $S_{A}$), i.e. the unique vector in $p\in\Delta(S_{A})$ satisfying
$pB_{S_{A}}=p$. Moreover, all entries of $p$ are positive.
\end{prop}
This specializes to Proposition \ref{prop:char} in case $S_{A}=S$.
In general, the consensus expectation still exists and is unique,
which is what we need for the examples of Section \ref{sec:optimism},
where irreducibility fails to hold. 

Our results do not hold if we relax our maintained finiteness assumption:
in Appendix \ref{subsec:Non-Existence} we report an example of \citet*{Hellman2011}
showing this. 

\subsection{Heterogeneous Coordination Weights and Their Relation to Self-Weights
in the Network\label{subsec:Heterogeneous}}

In the linear best-response game, we assumed that all agents put a
common weight $\beta$ on others' actions, and studied the limit $\beta\uparrow1$.
We again maintain the assumption that $\Gamma$ is irreducible and
consider now a more general class of environments, characterized by
$(\Gamma,\boldsymbol{\beta},y)$, where $\boldsymbol{\beta}=\left(\beta^{i}\right)_{i\in N}$
is a profile of agent-specific weights. In the coordination game associated
with such an environment, the linear best responses are given by 
\begin{equation}
a^{i}=(1-\beta^{i})E^{i}y+\beta^{i}{\displaystyle \sum\limits _{j\neq i}\gamma^{ij}E^{i}a^{j}}.\label{eq:coordination-game-different-beta}
\end{equation}
Paralleling our main study, we can ask what happens as $\beta^{i}\rightarrow1$
simultaneously across $i$. As we show in this section, this issue
is closely related to ``self-weights'' $\gamma^{ii}$ in the network. 

First, we consider some simple examples. Suppose $|N|=2$, with the
network 
\[
\Gamma=\left(\begin{array}{cc}
0 & 1\\
1 & 0
\end{array}\right).
\]
If $\beta^{2}=1$ and $\beta^{1}<1$, then in the limit $\beta^{1}\uparrow1$,
iterating the (modified) best-response equation $a^{i}=(1-\beta^{i})E^{i}y+\beta{\displaystyle ^{i}\sum\limits _{j\neq i}\gamma^{ij}E^{i}a^{j}}$
shows that we would have a convention given by\footnote{See Appendix \ref{periodicity-1} for discussion of such limits, called
\emph{simple }higher-order expectations, and related calculations.}
\[
\lim_{n\to\infty}\left[E^{1}E^{2}\right]^{n}E^{1}y.
\]
On the other hand, suppose $|N|=2$, $\beta^{1}=1$, and $\beta^{2}<1$.
Now, in the limit $\beta^{2}\uparrow1$, we would symmetrically have
a convention given by the (different)\ simple higher-order expectation
\[
\lim_{n\to\infty}\left[E^{2}E^{1}\right]^{n}E^{2}y
\]
Thus, with agent-specific self-weights\textemdash the corresponding
convention will depend on the details of how the limit $\left(\beta^{1},\ldots,\beta^{|N|}\right)\rightarrow\left(1,\ldots,1\right)$
is taken\textemdash in particular, whose $\beta^{i}$ converges faster
to $1$.

We briefly sketch how our analysis can be adapted to this case by
changing the network. In particular, we will show how the linear best-response
game with weights $\left(\beta^{1},\ldots,\beta^{|N|}\right)$ and
network $\Gamma$ has the same solution as the game with a common
coordination weight $\widehat{\beta}$ (that depends on $\left(\beta^{1},\ldots,\beta^{|N|}\right)$)
and an alternative network $\widehat{\Gamma}$. The alternative network
$\widehat{\Gamma}$ may, in general, have nonzero self-weights ($\widehat{\gamma}^{ii}>0$
for some $i$) even if $\Gamma$ has zero self-weights ($\gamma^{ii}=0$
for all $i$). The transformation applies to any $\Gamma$, with or
without positive entries on its diagonal.

Note that, in general, consensus expectations were defined allowing
the possibility of positive self-weights in $\Gamma$, and our analysis
of their basic properties (e.g., Propositions \ref{prop:char}, \ref{prop:representation}
and \ref{prop:cpa-signals}) applies in that case as well. Positive
self-weights are unnatural in some applications. For instance, in
the game with the interpretation that each agent is a single player,
one's best response cannot (by definition) depend on one's own action.
On the other hand, there are other applications where self-weights
have reasonable interpretations. For example, in the game with linear
best responses, if we replace each agent with a continuum of identical
agents (as we have done in the finance application), it would be natural
to think of an agent caring about the average action of individuals
like himself (i.e., in the same class). The same holds in the financial
trading application, assuming there is a possibility that a player
will sell into his own market (whose traders have the same expectation,
and thus the same interim beliefs). In those cases, characterizing
the average action of each population results in the equilibrium equations
we have been studying, but with positive entries permitted on the
diagonal of $\Gamma$. Formally, we can construct an analogue of the
game in Section \ref{subsec:game} and prove, paralleling part of
Fact \ref{fact:rationalizable}, that the game has a unique rationalizable
strategy profile. (The proof is by the same contraction argument used
to prove Fact \ref{fact:rationalizable}.)

To see what that strategy profile is, we describe the transformation
of any environment to one with an agent-independent common weight
on others' actions:
\begin{prop}
\label{prop:heterogeneous-beta}Given an environment with a network
$\Gamma$ and a vector $\boldsymbol{\beta}=(\beta^{i})_{i\in N}$,
define $\widehat{\beta}=\underset{i\in N}{\max}\beta^{i}$ and define
$\widehat{\Gamma}$ by
\[
\widehat{\gamma}^{ii}=\frac{\widehat{\beta}-\beta^{i}}{\widehat{\beta}\left(1-\beta^{i}\right)}\text{ and }\widehat{\gamma}^{ij}=\gamma^{ij}\left(1-\gamma^{ii}\right).
\]
For any $y$, the environments described by $(\widehat{\Gamma},\widehat{\beta},y)$
and $(\Gamma,\boldsymbol{\beta},y)$ have identical play in their
respective unique rationalizable strategy profiles.
\end{prop}

\subsubsection{Proof of Proposition  \ref{prop:heterogeneous-beta}}

By Fact \ref{fact:rationalizable}, there is a unique rationalizable
strategy profile in environment $(\widehat{\Gamma},\widehat{\beta},y)$.
In that strategy profile, player $i$'s action given his signal satisfies
\[
a^{i}=(1-\widehat{\beta})E^{i}y+\widehat{\beta}\sum\limits _{j}\widehat{\gamma}^{ij}E^{i}a^{j}.
\]
Splitting the $j=i$ term out of the last summation, and then using
the definition $\widehat{\gamma}^{ij}=\gamma^{ij}\left(1-\gamma^{ii}\right)$,
we have
\[
a^{i}=(1-\widehat{\beta})E^{i}y+\widehat{\beta}\widehat{\gamma}^{ii}E^{i}a^{i}+\widehat{\beta}\left(1-\widehat{\gamma}^{ii}\right){\displaystyle \sum\limits _{j\neq i}\gamma^{ij}E^{i}a^{j}\text{.}}
\]
Rearranging and using $E^{i}a^{i}=a^{i}$ gives 
\[
\left(1-\widehat{\beta}\widehat{\gamma}^{ii}\right)a^{i}=(1-\widehat{\beta})E^{i}y+\widehat{\beta}\left(1-\widehat{\gamma}^{ii}\right){\displaystyle \sum\limits _{j\neq i}\gamma^{ij}E^{i}a^{j}}
\]
and thus 
\begin{eqnarray}
a^{i} & = & \frac{1-\widehat{\beta}}{1-\widehat{\beta}\widehat{\gamma}^{ii}}E^{i}y+\frac{\widehat{\beta}\left(1-\widehat{\gamma}^{ii}\right)}{1-\widehat{\beta}\widehat{\gamma}^{ii}}{\displaystyle \sum\limits _{j\neq i}\gamma^{ij}E^{i}a^{j}}\nonumber \\
 & = & (1-\beta^{i})E^{i}y+\beta^{i}{\displaystyle \sum\limits _{j\neq i}\gamma^{ij}E^{i}a^{j}},\label{eq:original-system}
\end{eqnarray}
where in the last step we have deduced from the formula $\widehat{\gamma}^{ii}=\frac{\widehat{\beta}-\beta^{i}}{\widehat{\beta}\left(1-\beta^{i}\right)}$
the fact that $\frac{1-\widehat{\beta}}{1-\widehat{\beta}\widehat{\gamma}^{ii}}=1-\beta_{i}$
and $\frac{\widehat{\beta}\left(1-\widehat{\gamma}^{ii}\right)}{1-\widehat{\beta}\widehat{\gamma}^{ii}}=\beta_{i}$. 

Now (\ref{eq:original-system}) is an equilibrium of the coordination
game in environment $(\Gamma,\boldsymbol{\beta},y)$ (recall equation
(\ref{eq:coordination-game-different-beta}) defining that game),
and so by uniqueness of the rationalizable outcome, the proof is complete.

\subsection{Separability and Connection to \citet*{samet1998iteratedA}\label{subsec: samet} }

In Section \ref{sec:consensus-and-network}, we showed that\textemdash fixing
the information structure and network\textemdash there are strictly
positive \emph{pseudopriors} $\left(\lambda_{\bm{\pi},\Gamma}^{i}\right)_{i\in N}$
such that 
\begin{equation}
c(y;\bm{\pi},\Gamma)=\sum_{i}e^{i}\mathbf{E}^{\lambda_{\bm{\pi},\Gamma}^{i}}y.\label{eq:representation-1}
\end{equation}
At the same time, we made the observation\textemdash which we have
now made explicit in the subscripts of $\lambda^{i}$\textemdash that
those pseudopriors may depend on both the information structure $\bm{\pi}$
and the network $\Gamma$. We say an information structure $\bm{\pi}$
satisfies \emph{separability} if the pseudopriors depend \emph{only}
on the information structure: 
\begin{defn}
The information structure $\bm{\pi}$ satisfies \emph{separability}
if there exists a profile $(\lambda_{\bm{\pi}}^{i})_{i\in N}$ such
that, for every irreducible $\Gamma$, we have $\lambda_{\bm{\pi},\Gamma}^{i}=\lambda_{\bm{\pi}}^{i}.$ 
\end{defn}
When separability holds, the asymmetric information affects the consensus
expectation in an additively separable way, with each agent's pseudoprior
being weighted by his eigenvector centrality. Thus the incomplete
information and the network can be analyzed separately. Networks matter
only via the network centrality weights, and the information structure
$\bm{\pi}$ affects only the pseudopriors $\bm{\lambda}=(\lambda_{\bm{\pi}}^{i})_{i\in N}$. 

We can illustrate the failure of separability with an example building
on the one in Case II of Section \ref{subsec:Three-Illustrative-Cases}.
Suppose we have three agents arranged in a cycle as shown in Figure
\ref{fig:case-ii}, with each considering his counterclockwise neighbor
over-optimistic and his clockwise neighbor over-pessimistic. The network
$\Gamma$ in which all weight goes counterclockwise (i.e., $\gamma^{i,i-1}=1$
for all $i$, with indices read modulo 3) gives the maximum consensus
expectation. The network\textemdash call it $\Gamma^{\prime}$\textemdash in
which all weight goes clockwise (i.e., $\gamma^{i,i+1}=1$ for all
$i$, with indices read modulo 3) gives the minimum consensus expectation
given the beliefs. Note that all agents are symmetric in each network.
Thus, in both networks, by symmetry all agents have the same eigenvector
centrality. 

If separability held, then the two networks would have the same consensus
expectation: We have just said that the centralities are the same
across them, and that the information structure also remains the same
if we reverse the direction of each link in the network. Since in
fact the consensus expectation differs (indeed, differs as much as
possible) across the two networks, we have a failure of the separability
property. 

We have already given one sufficient condition for separability in
Section \ref{sec:cpa-signals}: a common prior on signals. Thus the
example described above cannot be consistent with a common prior on
signals. In \citet*{GolubMorris2016} we give a \emph{necessary} condition
for separability. We now informally report the condition in stages.
First, note that for higher-order expectations, and therefore consensus
expectations, the only beliefs about others that enter are \emph{marginal}
distributions over another's signal. An agent is never concerned about
the correlation in the signals of two or more others. This already
suggests that the common prior assumption on signals is more than
we need: Recall from Definition \ref{def:cpa-signals} that the common
prior assumption on signals places strong restrictions on beliefs
about \emph{profiles} of signals. In fact, separability is implied
by a weaker sufficient condition\textemdash one that requires priors
about signals to agree only in their marginals on every agent's signal.
Like the existence of a common prior on signals, such a property puts
no restrictions on agents' beliefs about $\Theta$ conditional on
signals, but it also relaxes substantially the restrictions on beliefs
about signals. 

In \citet*{GolubMorris2016} we show that an even weaker condition
is necessary and sufficient: We call it \emph{higher-order expectation-consistency}.
This condition specifies that we can find a ``pseudoprior'' for
each agent with the property that those pseudopriors have the same
expectations of all random variables in a certain class. The class
consists of all higher-order expectations of random variables that
are $\Theta$-measurable. In effect, this necessary and sufficient
condition imposes only those restrictions on higher-order beliefs
that are relevant to higher-order expectations. 

Our results in both this paper and \citet*{GolubMorris2016} relate
closely to and build on those of \citet*{samet1998iteratedA}. Samet
showed that\textemdash if one fixes a state space and agents' information
on that state space (modeled via a partitional information structure)\textemdash then
higher-order expectations of all random variables converge. If the
common prior assumption holds, they converge to ex ante expectations
under the common prior. Our Proposition \ref{prop:cpa-signals} is
a version of this result; critically, however, the reasoning is applied
not to the whole state space but to the space of signal profiles.
Samet also showed a converse: If all higher-order expectations of
any random variable converge to the same number (depending on the
random variable) regardless of the order in which they are taken,
then the information structure must satisfy the common prior assumption.
We do not have a converse in this paper. The characterization of the
separability result in \citet*{GolubMorris2016}, which we have described
above, is tight and thus is the closest analogue to \citet*{samet1998iteratedA}.
There are many conceptual and technical issues that distinguish our
notion of separability from the properties that matter in \citet*{samet1998iteratedA};
these differences are discussed in detail in \citet*{GolubMorris2016}. 

There is also another important technical and methodological connection
to \citet*{samet1998iteratedA}. We follow \citet*{samet1998iteratedA}
in representing information structures\textemdash as well as a network,
which we add to the model\textemdash via a Markov process. However,
we actually work with a different sort of Markov process than the
one in \citet*{samet1998iteratedA}: Our Markov process operates on
the \emph{union} of agents' types (which we denote by $S$), whereas
Samet's process applied to our questions operates on \emph{profiles}
of agents' types $T$.\footnote{Samet works with a partitional formalism; \citet*[Section 6]{GolubMorris2016}
restates our framework in that formalism.} There are a number of reasons why the former Markov process (on $S$)
is the appropriate one for our problem. First, it permits a unified
or symmetric treatment of networks and asymmetric information, as
discussed in Section \ref{subsec:Interpreting-the-Interaction-Structure}.
If one adds a network structure to Samet's Markov formulation, networks
and asymmetric information enter in very different ways in the formalism
(see \citet*{GolubMorris2016} for a presentation along these lines).
Second, and relatedly, our formalism allows us to relate key elements
of our analysis to results in the literature on network games. Finally,
the \citet*{samet1998iteratedA} approach works with matrices whose
rows and columns are indexed by $\Omega=\Theta\times\prod_{i\in N}T^{i}$,
which can be much larger than $S=\bigcup_{i\in N}T^{i}$; thus it
can be convenient to have our formalism for doing explicit computations.

\subsection{Ex Ante and Interim Interpretation\label{subsec:ex-ante-interim}}

We take an ex ante perspective in our analysis: At an initial date,
agents have prior beliefs\textemdash and no information\textemdash about
a state of the world. This interpretation entails common certainty
among the agents of everyone's prior beliefs and the way agents update
their beliefs.\footnote{Under an interim interpretation, it is without loss of generality
to assume common certainty of types' interim beliefs, i.e. how beliefs
are updated: see \citet*[p. 1237]{auma76} and \citet*{Brandenburger1993}.} In this section, we discuss some consequences of our ex ante approach
and interim interpretations of our results

\subsubsection{Dynamic Interpretation: The Arrival of Information }

Under the ex ante perspective, the results of this paper can be given
an explicitly dynamic interpretation. Before the arrival of information,
there is symmetric information and, therefore, the consensus expectation
is equal to the average of agents' ex ante expectations, weighted
by their eigenvector centralities (Section \ref{sec:consensus-and-network}).
In other words, if the agents had to select actions at that stage,
this is what their actions would be equal to. One interpretation of
our results is as an answer to the question, \emph{How does the consensus
expectation change after agents observe their signals?} We show that
common prior over signals is a sufficient condition for \emph{no}
change in the consensus expectation (Proposition \ref{prop:cpa-signals});
second-order optimism causes the consensus expectation to increase
to the highest possible interim belief (Proposition \ref{prop:ignorant-player});
and, under the conditions in the results on the tyranny of the least-informed,
the weights on agents' priors change from those induced by the network
$\Gamma$ to a degenerate vector which places all the weight on the
least informed.

\subsubsection{Interim Interpretation\label{subsec:Interim-Interpretation}}

Though we take an ex ante view throughout, consensus expectations,
which emerge from agents' play at the interim stage, cannot depend
on agents' ex ante beliefs about their own types. Thus consensus expectations
must depend only on agents' interim beliefs (across all possible types).
We have emphasized this in our notation, by first expressing the information
structure in interim terms (i.e., via the beliefs $\pi^{i}(\cdot\mid t^{i})$),
and only then adding in ex ante beliefs over each agent's signals
(the $\lambda^{i}$ in Proposition \ref{prop:representation}). 

Let us discuss how certain main results would look if we were to stick
to a purely interim interpretation. First, results such as the representation
of Proposition \ref{prop:representation} would still make sense,
but the $\lambda^{i}$ would not be interpreted as anyone's beliefs.
More substantially, consider Proposition \ref{prop:cpa-signals}.
Let us focus on a particularly simple consequence of it: Under the
common prior assumption on all of $\Theta\times T$, the consensus
expectation of $y$ is the prior expectation of $y$. To make sense
of this in interim terms, we first have to say what the common prior
assumption means in interim terms. \citet*{samet1998iteratedA} has
characterized that assumption as the conditition that, for any random
variable $y$, \emph{higher-order expectations converge to the same
number, independent of the order in which expectations are taken (as
long as each agent appears infinitely often)}; this number can be
identified with the common prior expectation of $y$. Thus an interim
statement of the simple consequence of Proposition \ref{prop:cpa-signals}
is: Under the italicized condition, the consensus expectation of $y$
is simply the prior expectation of $y$. This is natural: We can write
the consensus expectation as an average of higher-order expectations,
and the irreducibility of $\Gamma$ ensures that all agents appear
infinitely often in each of them. An interim version of Proposition
\ref{prop:cpa-signals} follows from very similar reasoning, with
more attention paid to the network, and this is carried out in \citet*{GolubMorris2016}.

The consequence of Proposition \ref{prop:cpa-signals} that we have
discussed is similar to Corollary \ref{cor:full-cpa} but differs
in an important way. Corollary \ref{cor:full-cpa} does not depend
on there being a common prior on the whole state space (i.e, on signals
and beliefs jointly); rather, it requires that the ex ante first-order
expectations of $y$ be the same across agents. This assumption does
not have an obvious interim interpretation. Thus, the contrast between
the ``full common prior'' result we have discussed in the previous
paragraph and the actual result of Corollary \ref{cor:full-cpa} helps
bring out where an ex ante perspective is important for us.

Our second-order optimism result (Proposition \ref{prop:over-optimism})
is stated in terms of interim beliefs only (the consensus expectation
is equal to the highest possible interim belief), so ex ante beliefs
do not play a role in the interpretation. On the other hand, the common
interpretation of signals property used in the result on the tyranny
of the least-informed (Proposition \ref{prop:ignorant-player}) does
not have any natural interim interpretation.\footnote{The ex ante properties of Definitions \ref{def: CIS_1} and \ref{def:CIS_2}
do imply properties of interim beliefs\textemdash see, for example,
Lemma \ref{lem:Bij-bound} in Section \ref{subsec:Proof-of-Proposition-Ignorant}.}

\subsection{Agent-Specific Random Variables and Incomplete Information about
the Network\label{heterogeneous-1}}

Our focus throughout the paper has been on agents' higher-order expectations
of a given random variable, which is the same across all agents. But
for many applications of interest, there is a different random variable
corresponding to each agent, and then higher-order expectations are
taken. For example, a literature on coordination games in networks
focuses on the case where agents have different preferred actions
(which correspond to the different random variables) in the absence
of coordination motives, and where one's network neighbors also influence
one's choice, with linear best responses assumed \citep*{Ballester2006,cadp15,behm15}. 

This case can be embedded readily into our formalism. Specifically,
suppose that instead of being interested in a (common) random variable
$y\in\mathbb{R}^{\Theta}$ measurable with respect to the external
state, each agent has a different random variable, $y^{i}\in\mathbb{R}^{\Theta}$.
Now, in Section \ref{sec:iterated-average-expectations}, equation
(\ref{eqn:def-first-step}) is changed to 
\[
x^{i}(1;\boldsymbol{y})=E^{i}y^{i}.
\]
Once $x^{i}(1;\boldsymbol{y})$ is set, the higher-order average expectations
are defined by the same equation, (\ref{eqn:def-iteration}), as before:
\[
x^{i}(n+1;y)=\sum_{j\in N}\gamma^{ij}E^{i}x^{j}(n;\boldsymbol{y}).
\]
Correspondingly, in the matrix notation of Section \ref{sec:The-Interaction-Structure},
where the key iteration is $x(n)={B}^{n-1}{F}y$, the vector $Fy$
is replaced by a vector $f\in\mathbb{R}^{S}$, with 
\[
f(t^{i})=E^{i}[y^{i}\mid t^{i}].
\]
The analogue of (\ref{eq:powers-Bn}) is
\[
\lim_{\beta\uparrow1}\left(1-\beta\right)\left({\displaystyle \sum\limits _{n=0}^{\infty}\beta^{n}B^{n}}\right)f\text{,}
\]
and $B^{n}f$ has the interpretation that it describes the higher-order
average expectations of the agents' first-order expectations of their
agent-specific random variables.

One can generalize further and consider a ``pure private values''
setting: $f$ can be replaced by an \emph{arbitrary }vector $f\in\mathbb{R}^{S}$,
with the interpretation that $f(t^{i})$ is the action that agent
$i$ would like to take, when he has signal $t^{i}$, in the absence
of coordination motives\textemdash an action he knows. In this case,
each agent faces no uncertainty about the random variable of interest
to him. Note that the case of different $y^{i}\in\mathbb{R}^{\Theta}$
is a special case of this, because in that case $f(t^{i})$ is agent
$i$'s \emph{expectation} of his own $y^{i}$ given signal $t^{i}$.

This brings us closer to \citet*{cadp15} and \citet*{behm14}. Motivated
by a study of endogenous attention allocation, they work with a network
version of a setting commonly studied in organizational economics
and focus on an analogue of our $\beta\uparrow1$ limit. They show
that it is a weighted average of agents' heterogeneous ideal points
that matters for determining the network consensus. Even though their
setting involves normally distributed random variables and linear\textendash quadratic
preferences, the core calculations boil down to understanding an analogue
of $B^{n}f$, just as we must in order to study the heterogeneous-values
variation of our model we have just presented.

\subsubsection{Equivalence Between Agent-Specific Random Variables and Different
Priors over $\Theta$}

Given any environment with agent-specific random variables and a common
prior on signals, we can find another environment with the same prior
on signals in which agents all care about the \emph{same} random variable
but have heterogeneous beliefs about external states. That is, given
any profile $(y^{i})_{i\in N}$, we can define a new environment with
new beliefs over $\Theta$ and a random variable $y$ so that the
resulting $f\in\mathbb{R}^{S}$ mimics that arising from the original
environment. Then results such as Proposition \ref{prop:cpa-signals}
can be applied.

This equivalence relies on the common prior on signals assumption:
Without a common prior on signals, we could maintain such an equivalence
only if the ``own random variables'' could depend on others' signals
\citep*[cf.][p. 74]{myerson1997}. This is related to the essential
differences we observed between the model with a common prior over
signals (Section \ref{sec:cpa-signals}) and the model without it. 

\subsubsection{Type-Dependent Network Weights\label{sec:type-dependent-gamma}}

A related extension allows for type-dependence in $\gamma^{ij}$.
In this case, we take this network weight to depend on the signal
of $i$, and write $\gamma^{ij}(t^{i})$. Much of our analysis goes
through unchanged: Equation (\ref{eqn:short-xn}) still describes
$x(n)$, but now under the definition
\[
B(t^{i},t^{j})=\gamma^{ij}(t^{i})\pi^{i}(t^{j}\mid t^{i}).
\]
If we interpret $\gamma^{ij}$ as $i$'s probability of meeting or
interacting with $j$, then signal-dependence of these weights corresponds
to private information about interactions. The only results that we
lose in this generalization are those of Section \ref{sec:cpa-signals},
because there is now no information-independent notion of the network
or of centrality. But the limits we study still exist, and much of
their structure (e.g., the structure described in Proposition \ref{prop:char},
with $p$ the left-hand unit eigenvector of the generalized $B$)
is still present and can be used to study this more general setting. 

\section{For Online Publication: Additional Discussion\label{sec:Additional-Discussion}}

\subsection{Periodicity and Simple Higher-Order Expectations\label{periodicity-1} }

In defining consensus expectations, or the limit of higher-order average
expectations, we considered the \emph{Abel average} 
\begin{equation}
\lim_{\beta\uparrow1}\left(1-\beta\right){\displaystyle \sum\limits _{n=0}^{\infty}\beta^{n}}x(n+1;y)=\lim_{\beta\uparrow1}\left(1-\beta\right)\left({\displaystyle \sum\limits _{n=0}^{\infty}\beta^{n}B^{n}}\right)Fy,\label{weighted-1}
\end{equation}
which is always well defined. It is natural to ask how the higher-order
average expectations $x(n;y)$ behave without this averaging, and
about the limit
\begin{equation}
\lim_{n\to\infty}B^{n}Fy.\label{unweighted-1}
\end{equation}
As long as $B$ is aperiodic,\footnote{A matrix is said to be aperiodic if, in the associated weighted directed
graph, the greatest common divisor of all cycles' lengths is equal
to $1$. A sufficient condition for this is that the matrix $\Gamma$
have all positive entries. Even if $\gamma^{ii}=0$ for all $i$\textemdash a
natural special case for some interpretations and applications\textemdash and
if there are at least $3$ agents, $\gamma^{ij}>0$ for all $j\neq i$
is another sufficient condition for aperiodicity.} the limit (\ref{unweighted-1}) exists and is equal to the right-hand
side of (\ref{weighted-1}). 

Aperiodicity, and the existence of the limit (\ref{unweighted-1}),
is not relevant for many of the applications reported in the paper.
\ For the linear best-response game and asset pricing, we are explicitly
interested in the limit of the weighted sum of higher-order expectations,
i.e., (\ref{weighted-1}) above, and not in limits of unweighted higher-order
expectations, i.e., (\ref{unweighted-1}) above. Nothing about the
structure of agent-type weights depends on aperiodicity.

However, periodicity does affect the behavior of the $x(n;y)$ in
the limit, and here we discuss how. Suppose that we have a cycle of
agents $i_{1},i_{2},\ldots,i_{|N|},i_{1}$: that is, that the network
$\Gamma$ has each agent $i_{k}$ putting weight $1$ on agent $i_{k+1}$.
Then the corresponding matrix will not be aperiodic. \ For example,
if there are two agents, $N=\left\{ 1,2\right\} $, and $\gamma^{12}=\gamma^{21}=1$,
then we have 
\[
\Gamma=\left(\begin{array}{cc}
0 & 1\\
1 & 0
\end{array}\right)
\]
and 
\[
B=\left(\begin{array}{cc}
0 & B^{12}\\
B^{21} & 0
\end{array}\right)
\]
(recall the definition of $B^{ij}$ from Section \ref{sec:The-Interaction-Structure})
and $B$ will be periodic and give rise to a two-cycle. In particular,
there will be well-defined limits \ 
\[
\lim_{n\to\infty}\left[E_{2}E_{1}\right]^{n}E_{2}y=c_{2}\mathbf{1}
\]
and 
\[
\lim_{n\to\infty}\left[E_{1}E_{2}\right]^{n}E_{1}\left(y\right)=c_{1}\mathbf{1}
\]
but they will not be equal. In the limit, the vector $x(n)$ will
cycle between 
\[
\left(\begin{array}{c}
c_{1}\mathbf{1}\\
c_{2}\mathbf{1}
\end{array}\right)\text{ and }\left(\begin{array}{c}
c_{2}\mathbf{1}\\
c_{1}\mathbf{1}
\end{array}\right).
\]

For more general cycles of agents, we will have limits of the form
\[
\lim_{n\to\infty}\left[E_{i_{1}}E_{i_{2}}...E_{i_{k}}y\right]^{n}E_{i_{1}}E_{i_{2}}...E_{i_{j}}y
\]
but they will be different for different values of $j=1,\ldots,k$.
\ We will refer to such expressions as \textit{simple higher-order
expectations}. \ If the network $\Gamma$ were given by this cycle,
then the entries of $B^{n}Fy$ would be the simple higher-order expectations.
The general higher-order expectations that we study will end up being
complicated weighted sums of such simple higher-order expectations,
although we will not in general work with the decompositions. 

Without the assumption of finitely many types, behavior more complicated
than cycling can arise, and ${\displaystyle (1-\beta)\sum\limits _{n=0}^{\infty}\beta^{n}}x(n+1;y)$
need not converge as $\beta\uparrow1$. This phenomenon is discussed
in \citet*{typicaltypes} and \citet*{NewT6:online}. A related but
different lack of convergence plays a role in \citet*{han16}: there,
because of the lack of finiteness of the type space, arbitrary higher-order
expectations can obtain.

\subsection{A Behavioral Interpretation of Irreducibility via No Trade\label{sec:no-trade}}

What is the behavioral content of the joint connectedness of beliefs
and the network, i.e., the irreducibility of $B$? 

We report a characterization of the joint connectedness property,
and therefore the existence and uniqueness of a distribution of positive
agent-type weights. Just as the common prior assumption can be characterized
as the non-existence of profitable trades among agents (see \citet*{morr94}
and \citet*{samet1998commonB}), the property we are studying here
has a no-trade characterization.\footnote{See \citet*{Nehring2001} for more on the various relations between
no-trade conditions, higher-order expectations, and common priors.} 

Let $x^{i}$ be a payment rule for agent $i$, $x^{i}:T^{i}\rightarrow\mathbb{R}$,
which is measurable with respect to agent $i$'s signal. \ A trade
consists of a profile of payment rules, $\left(x^{i}\right)_{i\in N}$.
The trade generates \textit{strict expected bilateral gains from trade}
if 
\[
x^{i}\left(t^{i}\right)\leq\sum\limits _{j\neq i}\gamma^{ij}{\displaystyle \sum\limits _{t^{j}\in T^{j}}\pi^{i}\left(t^{j}|t^{i}\right)x^{j}\left(t^{j}\right)}
\]
for each agent $i$ and $t^{i}\in T^{i}$, with strict inequality
for at least one agent $i$ and $t^{i}\in T^{i}$. The interpretation
is that agent $i$ is committed to making a payment $x^{i}\left(t^{i}\right)$
as a function of his signal. But he anticipates receiving the payments
to which others are committed. 
\begin{prop}
There exists a separable trade generating strict expected bilateral
gains from trade if and only if beliefs and the network are jointly
connected. 
\end{prop}
\begin{proof}
The existence of a separable trade giving strict expected bilateral
gains from trade is equivalent to the requirement that there exists
a vector $x$ such that $x>Bx$, where $>$ means a weak inequality
on all components and strict inequality on some component. Recall
that irreducibility implies the existence of a strictly positive \textit{vector
of agent-type weights} $p$ with $pB=p$. Now we have $px>pBx=pBx$,
a contradiction. So irreducibility fails. Conversely, suppose that
irreducibility fails. Then there exists at least one type $t^{i}\in S$
that no one assigns positive probability to, so that $\gamma^{ji}\pi^{j}\left(t^{i}\mid t^{j}\right)=0$
for all $\,t^{j}$. \ But now if we set $x^{i}\left(t^{i}\right)<0$
and $x^{j}\left(t^{j}\right)=0$ if $t^{j}\neq t^{i}$, then we have
a separable trade with strict expected gains.
\end{proof}

\subsection{Non-Existence of Consensus Expectations on Infinite State Spaces\label{subsec:Non-Existence}}

We have maintained the assumption that $\Theta$ and all the $T^{i}$
are finite. In general, without finiteness, there may not be a vector
of agent-type weights as defined in Proposition \ref{prop:char}.
An example offered by \citet*[Section 6]{Hellman2011} demonstrates
this. The example uses a version of the two-player information structure
in Rubinstein's \citeyearpar{rubi89} electronic mail game, with $T^{1}$
and $T^{2}$ both having the cardinality of $\mathbb{N},$ the natural
numbers. If we take a network $\Gamma$ on two players such that each
puts all weight on the other, and construct a suitable infinite analogue
of $B$, Hellman's result implies that there is no invariant measure
for $B$\textemdash i.e., no vector $p\in\Delta(S)$ of agent-type
weights such that $pB=p$. Therefore, there is no analogue of Proposition
\ref{prop:char}, which was the foundation for all our results.

We conjecture that if, like \citet*{Hellman2011}, we require the
state space $\Omega$ underlying $T^{1}$ and $T^{2}$ to be compact\footnote{Hellman works with a partitional formalism similar to that of \citet*{samet1998iteratedA};
see \citet*{GolubMorris2016} for a translation of higher-order expectations
into this framework.} and the information structure to be everywhere mutually positive
in his sense, then we can recover suitable analogues of our results.
On the other hand, if the states come from normal distributions and
agents receive noisy signals about them, then the relevant type spaces
are uncountably infinite and the random variables have unbounded support,
but iterated expectations can still be well-behaved. This case is
studied in \citet*{han16}; see also \citet*{mosh02}.
\end{document}